\documentclass[journal]{IEEEtran}

\IEEEoverridecommandlockouts

\usepackage{hyperref}

\usepackage[mathscr]{eucal}
\usepackage[cmex10]{amsmath}
\usepackage{epsfig,epsf,psfrag}
\usepackage{amssymb,amsmath,amsthm,amsfonts,latexsym}
\usepackage{amsmath,graphicx,bm,url}
\usepackage[dvipsnames]{xcolor}
\usepackage[caption=false]{subfig}
\usepackage{array}
\usepackage{verbatim}
\usepackage{bm}
\usepackage{algorithmic, cite}
\usepackage{algorithm}
\usepackage{verbatim}
\usepackage{textcomp}
\usepackage{mathrsfs}
\usepackage{enumerate}
\usepackage{epstopdf}

\usepackage{booktabs}
\usepackage{multirow}

\catcode`~=11 \def\UrlSpecials{\do\~{\kern -.15em\lower .7ex\hbox{~}\kern .04em}} \catcode`~=13

\allowdisplaybreaks[3]

\newcommand{\nn}{\nonumber}

\newcommand{\calA}{\mathcal{A}}

\newcommand{\calC}{\mathcal{C}}

\newcommand{\calI}{\mathcal{I}}

\newcommand{\calN}{\mathcal{N}}
\newcommand{\calO}{\mathcal{O}}
\newcommand{\calP}{\mathcal{P}}
\newcommand{\calQ}{\mathcal{Q}}

\newcommand{\calS}{\mathcal{S}}

\newcommand{\ba}{\mathbf{a}}
\newcommand{\bA}{\mathbf{A}}

\newcommand{\bd}{\mathbf{d}}
\newcommand{\bD}{\mathbf{D}}

\newcommand{\bh}{\mathbf{h}}

\newcommand{\bL}{\mathbf{L}}

\newcommand{\bo}{\mathbf{o}}
\newcommand{\bO}{\mathbf{O}}

\newcommand{\bq}{\mathbf{q}}
\newcommand{\bQ}{\mathbf{Q}}

\newcommand{\bS}{\mathbf{S}}

\newcommand{\bU}{\mathbf{U}}
\newcommand{\bv}{\mathbf{v}}

\newcommand{\bx}{\mathbf{x}}
\newcommand{\bX}{\mathbf{X}}

\newcommand{\bY}{\mathbf{Y}}

\newcommand{\bbE}{\mathbb{E}}

\newcommand{\bbP}{\mathbb{P}}

\DeclareMathAlphabet{\mathbsf}{OT1}{cmss}{bx}{n}
\DeclareMathAlphabet{\mathssf}{OT1}{cmss}{m}{sl}

\DeclareSymbolFont{bsfletters}{OT1}{cmss}{bx}{n}
\DeclareSymbolFont{ssfletters}{OT1}{cmss}{m}{n}
\DeclareMathSymbol{\bsfGamma}{0}{bsfletters}{'000}
\DeclareMathSymbol{\ssfGamma}{0}{ssfletters}{'000}
\DeclareMathSymbol{\bsfDelta}{0}{bsfletters}{'001}
\DeclareMathSymbol{\ssfDelta}{0}{ssfletters}{'001}
\DeclareMathSymbol{\bsfTheta}{0}{bsfletters}{'002}
\DeclareMathSymbol{\ssfTheta}{0}{ssfletters}{'002}
\DeclareMathSymbol{\bsfLambda}{0}{bsfletters}{'003}
\DeclareMathSymbol{\ssfLambda}{0}{ssfletters}{'003}
\DeclareMathSymbol{\bsfXi}{0}{bsfletters}{'004}
\DeclareMathSymbol{\ssfXi}{0}{ssfletters}{'004}
\DeclareMathSymbol{\bsfPi}{0}{bsfletters}{'005}
\DeclareMathSymbol{\ssfPi}{0}{ssfletters}{'005}
\DeclareMathSymbol{\bsfSigma}{0}{bsfletters}{'006}
\DeclareMathSymbol{\ssfSigma}{0}{ssfletters}{'006}
\DeclareMathSymbol{\bsfUpsilon}{0}{bsfletters}{'007}
\DeclareMathSymbol{\ssfUpsilon}{0}{ssfletters}{'007}
\DeclareMathSymbol{\bsfPhi}{0}{bsfletters}{'010}
\DeclareMathSymbol{\ssfPhi}{0}{ssfletters}{'010}
\DeclareMathSymbol{\bsfPsi}{0}{bsfletters}{'011}
\DeclareMathSymbol{\ssfPsi}{0}{ssfletters}{'011}
\DeclareMathSymbol{\bsfOmega}{0}{bsfletters}{'012}
\DeclareMathSymbol{\ssfOmega}{0}{ssfletters}{'012}

\newcommand{\Bin}{\mathrm{Bin}}

\newcommand{\bzero}{\mathbf{0}}

\newtheorem{theorem}{Theorem}
\newtheorem{lemma}[theorem]{Lemma}

\newtheorem{proposition}[theorem]{Proposition}

\newtheorem{remark}{Remark}

\newtheorem{data model}{Data Model}

\newcommand{\qednew}{\nobreak \ifvmode \relax \else
      \ifdim\lastskip<1.5em \hskip-\lastskip
      \hskip1.5em plus0em minus0.5em \fi \nobreak
      \vrule height0.75em width0.5em depth0.25em\fi}

\newcommand{\crossprod}{\times}

\newcommand{\prob} {\bbP}
\newcommand{\cond} {\:\vert\:}

\newcommand{\abs}[1] {\left| {#1} \right|}
\newcommand{\norm}[1] {\left\Vert{#1}\right\Vert}

\newcommand{\bracket}[1] {\left[ {#1} \right]}
\newcommand{\curly}[1] {\left\{ {#1} \right\}}

\newcommand{\bigO}{\mathcal{O}}

\newcommand{\softO}{\tilde{\calO}}

\newcommand{\bigOmega}{\Omega}
\newcommand{\littleOmega}{\omega}
\newcommand{\softOmega}{\tilde{\Omega}}

\newcommand{\bigTheta}{\Theta}
\newcommand{\softTheta}{\tilde{\Theta}}

\def\nminp{n'_{min}}
\def\nmin{n_{min}}
\def\nminbound{b}

\def\mumin{\mu_{min}}

\newcommand{\Atyp}[1] {\calA_{#1}}
\newcommand{\InAtyp}[1] {\bA \in \Atyp{#1}}
\newcommand{\NotInAtyp}[1] {\bA \notin \Atyp{#1}}

\newcommand{\colsampprob} {s}
\newcommand{\clustsampprob} {t}

\newcommand{\colinner} {U}
\newcommand{\colinnernorm} {u}

\newcommand{\samprand} {V}
\newcommand{\samp} {v}

\newcommand{\clusteridx} {\calC}

\newcommand{\AOneFull} {\overline{\calA}_1}
\newcommand{\AOneObs} {\calO_1}

\newcommand{\fullobsmat} {\overline{\bA}}
\newcommand{\fullobsvec} {\overline{\ba}}
\newcommand{\fullobsel} {\overline{a}}
\newcommand{\obsmat} {\bO}
\newcommand{\obsvec} {\bo}
\newcommand{\obsel} {o}

\newcommand{\InAOneFull} {\AOneFull}
\newcommand{\NotInAOneFull} {\AOneFull^C}

\newcommand{\obsset} {\calO}

\newcommand{\balRS}{f}
\newcommand{\balSbS}{g}
\newcommand{\densdiff}{\gamma}
\newcommand{\obsprob}{\rho}
\newcommand{\densavg}{\tau}

\newcommand{\suffepsA}{\alpha}
\newcommand{\suffepsB}{\beta}
\newcommand{\pqprimeeps}{\epsilon_1}
\newcommand{\rhoprimeeps}{\epsilon_2}

\newcommand{\blockidx}[1] {\calI_{#1}}

\def\smallq/{small-q}
\def\largeq/{large-q}

\newcommand{\SbsNminpConst}{\zeta'}

\usepackage{multirow}

\begin{document}
\title{Scalable and Robust Community Detection With Randomized Sketching}

\author{
Mostafa~Rahmani,
Andre Beckus,~\IEEEmembership{Student Member,~IEEE},
Adel Karimian, and
George~K.~Atia,~\IEEEmembership{Senior Member,~IEEE}
\thanks{This work was supported by NSF CAREER Award CCF-1552497.
The University of Central Florida Advanced Research Computing Center provided computational resources that contributed to results reported herein.

M. Rahmani, A. Beckus, and G. K. Atia are with the Department of Electrical and Computer Engineering, University of Central Florida, Orlando, FL 32816 USA.

A. Karimian, deceased, was with the Department of Electrical and Computer Engineering, University of Central Florida, Orlando, FL 32816 USA.

M.~Rahmani and A. Beckus contributed equally to this work.

A conference version of this work was presented at the 52nd Annual Asilomar Conference on Signals, Systems, and Computers, 2018.
}
}

\markboth{}
{Shell \MakeLowercase{\textit{et al.}}: Bare Demo of IEEEtran.cls for Journals}

\maketitle

\begin{abstract}
This article explores and analyzes the unsupervised clustering of large partially observed graphs. We propose a scalable and provable randomized framework for clustering graphs generated from the stochastic block model. The clustering is first applied to a sub-matrix of the graph's adjacency matrix associated with a reduced graph sketch constructed using random sampling. 
Then, the clusters of the full graph are inferred based on the clusters extracted from the sketch using a correlation-based retrieval step.
Uniform random node sampling is shown to improve the computational complexity over clustering of the full graph when the cluster sizes are balanced.
A new random degree-based node sampling algorithm is presented which significantly improves upon the performance of the clustering algorithm even when clusters are unbalanced.
This framework improves the phase transitions for matrix-decomposition-based clustering with regard to computational complexity and minimum cluster size, which are shown to be nearly dimension-free in the low inter-cluster connectivity regime.
A third sampling technique is shown to improve balance by randomly sampling nodes based on spatial distribution.
We provide analysis and numerical results using a convex clustering algorithm based on matrix completion.
\end{abstract}


\IEEEpeerreviewmaketitle

\section{Introduction}
\label{sec:intro}

The identification of clusters within graphs constitutes a critical component of network analysis and data mining, and can be found in a wide array of practical applications ranging from social networking \cite{Mishra:2007:CSN:1777879.1777884} to biology \cite{2016arXiv160904316N}.
Community detection algorithms identify communities or clusters of nodes within which connections are more dense (see \cite{FORTUNATO20161} and references therein).
Many algorithms have been proposed, including classic approaches such as spectral clustering \cite{von2007tutorial}, as well as newer approaches incorporating random walks \cite{pons2005computing} or solution of semidefinite programming problems \cite{hajek2016semidefinite}.

One issue found in many community detection algorithms is the high computational cost, e.g., by requiring the exploration of the entire graph, and storage of the full adjacency matrix in memory.
Various methods have been developed to avoid these issues, for example by utilizing local optimization heuristics with efficient data structures \cite{PhysRevE.70.066111}, or by developing distributed algorithms which reduce the computational and storage burden placed on individual compute nodes \cite{pmlr-v37-yange15}.
We propose a randomized framework which enables efficient processing of large graphs by a different means.
Our framework leverages the intrinsic low-dimensionality of the problem so that the costly clustering operation is only applied to a small subgraph obtained via sampling.
By doing so, existing algorithms can be invoked with improved scalability, thus enabling their use in processing larger graphs than they could otherwise handle.

Another recurring issue is that many algorithms fail to identify small clusters. Here, we propose an algorithm which can handle graphs that are highly unbalanced, i.e., when the cluster sizes are highly disproportionate.

To compare with previous work, we perform analysis and experiments with the Stochastic Block Model (SBM) \cite{HOLLAND1983109,CondonKarp:01}, along with a modification to incorporate partial observations. 
Graphs created from this probabilistic generative model contain a planted set of disjoint clusters (where each node must belong to a cluster). Edges within each cluster are created with probability $p$, and edges between clusters exist with probability $q$.
Any given edge is observed with probability $\obsprob$.
In this work, we establish conditions for exact recovery, i.e., where the probability that the algorithm \emph{exactly} returns the correct planted partition approaches one as the number of nodes $N$ increases.

\subsection{Contributions}
We propose an approach in which a graph clustering algorithm is applied to
a small random subgraph obtained from the original graph through random sampling of the nodes.
We provide an analysis establishing conditions for exact recovery using the proposed approach.
The approach allows flexibility to choose both the graph clustering algorithm and the randomized sampling technique.
Here, we perform the analysis using the low rank plus sparse matrix decomposition clustering technique described in \cite{Chen:2014:CPO:2627435.2670322,NIPS2014_5309}.
Three randomized techniques are proposed, each of which varies in the information required from the full graph. 

Uniform Random Sampling (URS) forms the sample set without any knowledge of the full graph, other than its size.
This simple approach can significantly improve the performance of the clustering, both in terms of computational complexity and memory requirements.
Let $N$ be the number of nodes in the full graph, $r$ be the number of clusters, and $N'$ be the sketch size used by the proposed algorithm, i.e. the number of nodes sampled from the full graph.
Suppose the data is balanced, i.e., the clusters are of size $\bigTheta(N)$\footnote{We say $f = \bigTheta(g)$ if $f = \bigO(g)$ and $f = \bigOmega(g)$.}.
If the edge and observation probabilities in the generative model are constant, then successful clustering can occur with high probability (whp) using a sketch of only $N' = \softO(r^2)$ nodes\footnote{The soft-O notation $\softO(\cdot)$, soft-$\Omega$ notation $\softOmega(\cdot)$, and soft-$\Theta$ notation $\softTheta(\cdot)$ ignore log factors.
Specifically, $f = \softO(g)$, $f = \softOmega(g)$, or $f = \softTheta(g)$ if there exists an $a \in \mathbb{R}$ such that $f(N) = \calO(g(N) \log^a(N))$, $f(N) = \Omega(g(N) \log^a(N))$, or $f(N) = \bigTheta(g(N) \log^a(N))$, respectively.
}.
This reduces the per-iteration computational complexity of the  costly clustering step from $\bigO \left( r N^2 \right)$ to only $\softO \left( r^3 \right)$.
This is a substantial reduction when the number of clusters scales sub-linearly with graph size, with the complexity becoming almost dimension-free when $r$ is order-wise constant.
However, when the clusters are unbalanced in size -- a challenging case in community detection -- a considerable number of samples will still be required. In particular, if the smallest cluster size is $\softTheta(\sqrt{N})$ then $\bigOmega(N)$ samples are required.

To address this issue, we study a sampling method in which the nodes are sampled with probability inversely proportional to their node degrees. This is equivalent to sampling based on the sparsity levels of the columns of the adjacency matrix, hence the appellation `Sparsity-based Sampling' (SbS).
By capturing smaller clusters and producing more balanced sketches, the clustering algorithm can be made even more likely to succeed with the sketch than with the full matrix.
Again, if the clusters are of size $\softOmega(\sqrt{N})$ and the inter-cluster edge probability $q$ diminishes sufficiently fast as $N$ increases, then clustering using SbS is highly likely to succeed, while requiring as few as roughly $\softTheta \left( r^2 \right)$ samples (see Section~\ref{sec:Sbs} for details).
This means that the number of samples is almost independent of the graph size.
In fact, this result holds under the same conditions even if the minimum cluster size is reduced to $\nmin = \softTheta(r)$ with a mild restriction on the number of clusters $r=\softO( \sqrt{N} )$.
This approach comes close to state-of-the-art performance.
As a point of comparison, under the same conditions, the best results  we are aware of require $\nmin = \softOmega(1)$ \cite{Chen:2016:STP:2946645.2946672,cai2015}.
Therefore, our approach comes close to state-of-the-art performance in this regime, while significantly reducing sample complexity to avoid the heavy cost of clustering a large graph.

Finally, we leverage a randomized structure-preserving sketching technique termed Spatial Random Sampling (SRS) \cite{7968311}. With this approach, the individual edges of each node are considered, and sampling is performed based on the spatial distribution of the column vectors in the adjacency matrix.
Numerical results show that this approach exceeds the performance of URS and SbS in many cases.

A table summarizing the results of this paper is shown in Table~\ref{tab:AsymptoticResults}. Additional symbols are defined in the caption with more explanation given in later sections. In all cases, the computational complexity is $\bigO(r N'^2)$ per iteration of the convex optimization algorithm.
\begin{table}
\centering
\caption{
\bf Exact recovery sufficient conditions for techniques described in this paper. Symbols are number of clusters $r$, intra-cluster edge probability $p$, inter-cluster edge probability $q$, observation probability $\rho$, density difference $\densdiff=1-2\max\{1-p,q\}$, and balance $\balRS=N / \nmin$.
The proposed algorithm builds a sketch with $N'$ sampled nodes.
}
\begin{tabular}{| c | c | c |}
\hline
Approach & Minimum cluster size $\nmin$ & Required samples $N'$ \\
 \hline
  Full Graph \cite{Chen:2014:CPO:2627435.2670322}
  & $\bigOmega \left( \frac{ \sqrt{N} \log N }{\sqrt{\obsprob} \densdiff} \right)$
  & $N$
 \\\hline
  URS
  & $\bigOmega \left( \frac{\sqrt{N} \log N }{\obsprob \densdiff^2} \right)$
  & $\bigOmega\left( \frac{ f^2 \log^2 N }{\obsprob \densdiff^2} \right)$
 \\\hline
  SbS (\largeq/)
  & \multirow{2}{*}{$\bigOmega \left( \frac{ r \sqrt{q N} \log^2 N}{\obsprob \densdiff^2} \right)$}
  & \multirow{2}{*}{$\bigOmega \left( \frac{r^2 q^2 f^2 \log^4 N}{\obsprob \densdiff^2} \right)$}
 \\
  $q = \littleOmega \left( \balRS^{-1} \right)$ & &
 \\\hline
  SbS (\smallq/)
  & \multirow{2}{*}{$\bigOmega \left( \frac{r \log^3 N}{\obsprob \densdiff^2} \right)$}
  & \multirow{2}{*}{$\bigOmega \left( \frac{r^2 \log^4 N}{\obsprob \densdiff^2} \right)$}
 \\
  $q = \bigO \left( \balRS^{-1} \right)$ & & 
 \\\hline
\end{tabular}
\label{tab:AsymptoticResults}
\end{table}

\noindent\textbf{Paper organization:}
In Section~\ref{sec:BackgroundRelated}, we provide background on the problem, and discuss related work.
In Section~\ref{sec:ProposedURS}, we present the structure of the proposed approach, which consists of node sampling, subgraph clustering, and full data clustering, along with its theoretical analysis. For ease of exposition, node sampling is first performed using URS.
In Section~\ref{sec:ProposedSbSSRS}, the new SbS sampling method for capturing small clusters is presented and analyzed.
Then, a second node sampling technique is presented.
We discuss the computational complexity of these methods in Section~\ref{sec:ComputComplex}.
Numerical and Experimental Results are shown in Section~\ref{sec:numericalresults}.
We conclude in Section~\ref{sec:conclusion}.

\section{Background and Related Work} \label{sec:BackgroundRelated}

\subsection{Sampling}

Randomized sketching techniques have been instrumental in devising scalable solutions to many high-dimensional unsupervised learning problems, such as robust Principal Component Analysis (PCA)  \cite{rahmani2017high,mackey2011divide,Candes:2011:RPC:1970392.1970395}, outlier detection \cite{rahmani2015randomized,lamport8},  low rank matrix approximation \cite{lamport14,gu1996efficient,nguyen2009fast}, data summarization \cite{rahmani2017spatial,lamport14,ailon2006approximate}, and data clustering \cite{ben2007framework,czumaj2004sublinear}.

The use of sampling naturally extends to graphs as well.
Sampling may be performed in the edges or the nodes, depending on which imposes the dominant cost.
We focus on node sampling (although our model accommodates partial observations, which can be considered as a form of uniform edge sampling).
In \cite{Maiya:2010:SCS:1772690.1772762}, a complexity-reducing scheme is explored, in which nodes are incrementally added to the sample such that at each step the number of nodes adjacent to the sample are maximized. However, this is largely experimental work, and conditions for exact recovery are not analyzed.
The approach in \cite{voevodski2010efficient} also uses random node sampling, but considers entirely different scenarios in which the data feature matrices are known.  Our setup only assumes availability of similarity and dissimilarity information given the partially observed topology, otherwise no data features are available. 
Other works have used node sampling to produce representative sketches of graphs, although not in the specific context of community detection, e.g. \cite{4781123,Leskovec:2006:SLG:1150402.1150479}.
Edge-based approaches include \cite{7511156}, where reduced execution time is demonstrated experimentally by performing uniform edge sampling, and \cite{yun2014communityhu} which demonstrates a reduced edge sampling budget through the use of either URS or active sampling.
These results are altogether different from our node-based approach.

The object of most sampling algorithms is to capture certain features of the full graph in the sketch, for example cut size \cite{7511156} or degree distribution \cite{4781123}.
Similarly, we aim to maintain the SBM features of the full graph in the sketch.
However, one feature which the SbS and SRS sampling algorithms \emph{intentionally seek to modify} is the proportion of the clusters in the underlying planted partition. In particular, the clusters in the sketch should be of equal size, or as close as possible to this ideal. By doing so, the sufficient conditions for successful recovery can be significantly relaxed.
The goal of obtaining a balanced sketch is also considered in \cite{doi:10.1063/1.4712602}.  However, they do  not consider this in the context of clustering, and no analysis or guarantees are provided.
In \cite{JMLR:v16:ailon15a}, the clustering algorithm iteratively finds and removes large clusters to improve the balance of the remaining small clusters. However, their technique does not improve the computational complexity as the full graph still needs to be clustered initially. Additionally, although they incorporate sampling, it is in the form of edge sampling via partial observations, rather than node sampling.

\subsection{Community Detection Algorithms}

Although there exists a vast array of algorithms for community detection, here we will focus on algorithms that have been analyzed with the SBM.
The notable work in \cite{959929} proposed a spectral algorithm and provided rigorous recovery conditions using perturbation analysis.
More recently, \cite{PhysRevE.84.066106} used physics-based arguments to conjecture a simple necessary and sufficient condition for weak recovery (i.e., where there is a positive correlation between the estimated and actual communities whp) for the regime where $p,q = \bigO(N^{-1})$. In the two-cluster case, the sufficient condition of the conjecture was proven by \cite{Mossel2017,Massoulie:2014:CDT:2591796.2591857} and the necessary condition by \cite{Mossel2015}.

Many clustering algorithms optimize a global objective function over the entire graph.
One approach, correlation clustering, explicitly minimizes the total number of disagreements between the actual graph and a partitioning of the graph into disjoint cliques \cite{Bansal2004}.
An advantage of this formulation is that it can be solved without prior knowledge of the number of clusters.
We work with a recent convex algorithm which solves this problem by decomposing the graph's adjacency matrix as a sum of low rank and sparse components \cite{Chen:2014:CPO:2627435.2670322,NIPS2014_5309}.
The cluster structure is captured by the low rank component in the form of non-overlapping full cliques, and the sparse component indicates missing edges within clusters and extra edges across clusters.
The validity of the low rank plus sparse structure emerges from the fact that $p \ge \frac{1}{2} \ge q$.
The decomposition provides a good foundation for analyzing the sampling technique. The SBM-based analysis in \cite{Chen:2014:CPO:2627435.2670322} provides a direct tradeoff between the SBM parameters: minimum block size, $p$, $q$, and $\obsprob$. The algorithm therein, if successful, is guaranteed to produce not only the disagreement minimizer, but also the correct clustering whp.

For comparison, we employ another class of convex optimization approaches which rely on relaxed Maximum Likelihood Estimators (MLEs).
One algorithm, \cite{cai2015} which we refer to as (Cai,2015) in the tables and figures, uses the alternating direction method of multipliers (ADMM) algorithm. This algorithm can handle more general graphs than the SBM, in that it accommodates outlier nodes. However, one disadvantage of this algorithm is that it requires that the number of clusters be known or estimated.
Another approach, \cite{NIPS2016_6574} referred to as (Jalali,2016), proposes a different convex optimization problem from \cite{cai2015}.
Critically, it requires strong apriori knowledge of the structure of the graph.
Specifically, it requires that the quantity $\sum_{k=1}^r n_k^2$ be known, where $n_k$ is the size of the $k\textsuperscript{th}$ cluster.
No algorithm is provided, but the problem can be solved using standard techniques.
The work of \cite{Chen:2016:STP:2946645.2946672}, which we call (Chen,2016), uses an optimization problem which is a special case of that in \cite{NIPS2016_6574}, in that it assumes all clusters are of the same size.

The algorithms of \cite{cai2015,NIPS2016_6574,Chen:2016:STP:2946645.2946672} do not explicitly address partial observations.
For these algorithms, we can simply treat the unobserved entries as having no edge.
The optimization problem of \cite{Chen:2014:CPO:2627435.2670322}, on the other hand, directly exploits knowledge of which graph entries are observed.
We note, however, that the algorithm does not exploit any additional side information beyond the observation status of each edge.
When side information is available, it can be used to significantly improve results.
For example, in \cite{2018arXiv180908353I}, tensor data is used as side information to improve community detection in partially observed coupled graph-tensor factorization (CGTF) models.

\subsection{Performance and Limitations of Existing Algorithms} \label{sec:LimitsAlg}

One factor in determining the success of an algorithm is the density difference $\densdiff=1-2\max\left\{1-p,q\right\}$ of the SBM. As this gap decreases, the intra-cluster and inter-cluster edges become harder to distinguish.  Another factor is the minimum cluster size $\nmin$; as the smallest cluster(s) get proportionally smaller, they become easier to ``lose'' in the noise. Here, the value of $q$, the density of noisy edges found between clusters, also plays an important role.

Assuming that other parameters do not scale with $N$, it is typically considered sufficient to have $\nmin=\softOmega(\sqrt{N})$.
This limit appears in the analytic sufficient conditions of many algorithms (see \cite{6873307} for a review of algorithms having this limit).
However, when considering some special cases, a more nuanced picture emerges.
It is shown in \cite{Chen:2016:STP:2946645.2946672} that for the case of equal-sized clusters, if $q = \bigO\left( \frac{\log^2 N}{N} \right)$ then a polynomial-time algorithm can achieve exact recovery whp for $\nmin=\Omega(\log N)$.
A similar result can be found in \cite{NIPS2016_6574} for the more general case of arbitrarily-sized clusters.
The authors of \cite{NIPS2016_6574} propose and analyze another polynomial-time convex algorithm which performs well under a wide variety of cluster scalings. In particular, they provide an example where clusters are as small as $\bigTheta ( \sqrt{\log N} )$, although this requires additional conditions on the size and number of clusters. Recall that this algorithm also requires additional side information in the form of the sum of squares of the cluster sizes.
A table comparing the cluster size lower bounds of this paper to other works is shown in Table~\ref{tab:AsymptoticComparison}.
For an extensive list of algorithms and sufficient conditions, see Table 1 of \cite{Chen:2014:CPO:2627435.2670322}.
Though the algorithms listed in the table are polynomial-time algorithms, at each iteration they all require a Singular Value Decomposition (SVD), which scales quadratically with the graph size.
A distinguishing feature of our work is that it can reduce the time of the SVD while coming close to state-of-the-art minimum cluster size requirements.

The authors of \cite{JMLR:v16:ailon15a} demonstrate that even if capturing the small clusters is difficult, finding the large clusters is typically easy. By finding and removing the large clusters, the proportions of the small clusters become more favorable for successful recovery. However, their iterative algorithm still needs to be run initially on the full graph, affording no computational advantage, and they do not provide sufficient conditions for exact recovery of an entire graph (although a weaker result is provided). This idea of improving the proportions of the graph can be found in our approach as well, but in our case the costly clustering step is only performed once, and performed on the sketch rather than the full graph.

\begin{table}
\centering
\caption{\bf Minimum cluster size requirements in literature.  For purposes of comparison, we assume density gap $p-q=\densdiff$. For the algorithms of \cite{Chen:2016:STP:2946645.2946672,cai2015,NIPS2016_6574}, unobserved entries are treated as a zero, so we make the substitutions $p \rightarrow \rho p$ and $q \rightarrow \rho q$.
}
\begin{tabular}{| c| c | c | }
\hline
Algorithm & $\nmin$ (\largeq/) & $\nmin$ (\smallq/) \\
 \hline
  (Chen,2014) \cite{Chen:2014:CPO:2627435.2670322}
  & $\bigOmega \left( \frac{ \sqrt{N} \log N }{\sqrt{\obsprob} \densdiff} \right)$
  & (same as \largeq/)
 \\\hline
  (Cai,2015) \cite{cai2015}
  & \multirow{2}{*}{$\bigOmega \left( \frac{\log N}{\obsprob \densdiff^2} \vee \frac{\sqrt{q N}}{\sqrt{\obsprob}\densdiff} \right)$}
  & \multirow{2}{*}{$\bigOmega \left( \frac{\log N}{\obsprob \densdiff^2} \right)$}
 \\
  (Chen,2016) \cite{Chen:2016:STP:2946645.2946672} & &
 \\\hline
  (Jalali,2016) \cite{NIPS2016_6574}
  & \!\!$\bigOmega \left( \frac{\log N}{\obsprob \densdiff^2} \vee \frac{ \sqrt{n_{max}} \vee \sqrt{q N} }{\sqrt{\obsprob} \densdiff} \right)$\!\!
  & \!\!$\bigOmega \left( \frac{\log N}{\obsprob \densdiff^2} \vee \frac{ \sqrt{n_{max}} }{\sqrt{\obsprob} \densdiff} \right)$\!\!
 \\\hline
  This paper (URS)
  & $\bigOmega \left( \frac{\sqrt{N} \log N }{\obsprob \densdiff^2} \right)$ & (same as \largeq/)
 \\\hline
  This paper (SbS)
  & $\bigOmega \left( \frac{r \sqrt{q N} \log^2 N}{\obsprob \densdiff^2} \right)$
  & $\bigOmega \left( \frac{r \log^3 N}{\obsprob \densdiff^2} \right)$
 \\\hline
\end{tabular}
\label{tab:AsymptoticComparison}
\end{table}

\subsection{Data Model}

The adjacency matrix is assumed to follow a variant of the Planted Partition/Stochastic Block Model \cite{HOLLAND1983109,CondonKarp:01}, which allows for partial observations.
This data model is defined as follows \cite{Chen:2014:CPO:2627435.2670322}.

\begin{data model}
The graph consists of $N$ nodes partitioned into $r$ clusters.  Any two nodes within a cluster are connected with probability $p$, and two nodes belonging to different clusters are connected with probability $q$.  Any given edge is observed with probability $\obsprob$. Given an adjacency matrix $\bA$, the decomposition into the symmetric low rank $\bL$ and sparse $\bS$ matrices takes the form $\bA \!=\! \bL + \bS$.
Matrix $\bL$ captures the connectivity of nodes within a cluster and has no inter-cluster edges, while $\bS$ indicates missing intra-cluster and extra inter-cluster edges.
\label{datamodel}
\end{data model}

We note that for an adjacency matrix $\bA$ with ideal cluster structure, i.e., where $p=1, q=0$, and $\obsprob = 1$, we have $\bL=\bA$ and $\bS=\bzero$.
The diagonal elements of $\bA$ are assumed to be fully observed and set to all ones for convenience.

\subsection{Notation}
Vectors and matrices are denoted using bold-face lower-case and upper-case letters, respectively.
Given a vector $\ba$, its $\ell_p$-norm is denoted by $\norm{\ba}_p$ and $a(i)$ its $i\textsuperscript{th}$ element.
Given matrix $\bA$, its $\ell_1$-norm (the sum of the absolute values of its elements) is denoted by $\norm{\bA}_1$, its nuclear norm (the sum of its singular values) is denoted by $\norm{\bA}_*$, and its $i\textsuperscript{th}$ column is denoted by $\ba_i$.
The operator $\Omega_{\text{obs}}(\cdot)$ returns the observed values of its vector or matrix argument.

\section{Structure of the Proposed Approach} \label{sec:ProposedURS}

Algorithm~\ref{alg:mainalgorithm} shows the proposed approach.
It consists of three main steps: node sampling, subgraph clustering, and full data clustering.
We will separately describe the three main steps of the algorithm.
In this section, we present the simplest sampling method, URS, to clearly illustrate the structure of the framework.
Analysis will demonstrate the capability of this method to improve computational complexity of the clustering algorithm, but will also reveal the weakness of this sampling method when presented with graphs containing small clusters.
Two improved sampling methods, SbS and SRS, will be presented in Section~\ref{sec:ProposedSbSSRS}.
Proofs are deferred to Appendix~\ref{sec:ProofSRS}.

\subsection{Node Sampling}
A key advantage to the proposed approach is that it does not apply the clustering algorithm to the full data, but rather to a sketch, i.e., a subgraph generated from a set of $N'$ randomly sampled nodes.
We denote the adjacency matrix of the sketch by $\bA' \in \mathbb{R}^{N' \times N'}$.
The parameter $\balRS=N/\nmin$ indicates how balanced the graph $\bA$ is. When $f=r$, the data is perfectly balanced, and larger values of $f$ indicate a lack of balance.

With URS, we sample each node with equal probability, and so the probability of sampling from the smallest cluster is exactly equal to $\balRS^{-1}$.
Therefore, we can expect that any imbalance in the full matrix will carry over to the sketch, thus requiring more samples to ensure enough columns are sampled from the smallest cluster.
The following lemma shows that the sufficient number of random samples to obtain at least $\nminbound$ sampled nodes from each cluster whp is in fact linear with $\balRS$.
\begin{lemma} [Sampling Size for URS] \label{lem:Suff_Optimality_Partial}
Suppose that $\bA'$ is produced by Step 1 (sampling) of Algorithm~\ref{alg:mainalgorithm} using Random Sampling of $N'$ columns.
Given an arbitrary constant $\nminbound \in (0,\nmin]$, if
\begin{align} \label{eqn:Suff_Optimality_Partial_RS_uniform}
N
\ge N'
\ge 2 \balRS \left[\nminbound + \log \left( 2 r N \right)\right],
\end{align}
then $\nminp > \nminbound$ with probability at least $1-N^{-1}$, where $\nminp$ is the size of the smallest cluster in the sketch matrix.

\end{lemma}

\subsection{Robust Subgraph Clustering}
Next, we solve \eqref{eq:convex_sketch} by applying the low rank plus sparse decomposition algorithm of \cite{Chen:2014:CPO:2627435.2670322} to cluster the sketch $\bA'$.
This approach is valid since the sketch can be decomposed as $\bA' \!=\! \bL' \!+\! \bS'$, where $\bL'$ and $\bS'$ are the sketches of $\bL$ and $\bS$, respectively (constructed using the same index set $\calI$).
Algorithm 1 of \cite{Chen:2014:CPO:2627435.2670322} can be used to search for improved solutions. This algorithm starts with initial value $\lambda = \frac{1}{32 \sqrt{N' \overline{\obsprob}}}$ and iteratively solves \eqref{eq:convex_sketch} while varying $\lambda$ until a valid solution is obtained.
Note that convergence is not guaranteed if this algorithm is invoked, and theoretical results assume it is not used.
Lemma \ref{lm:completion} provides a sufficient condition for this step to exactly recover $\bL'$.
\begin{lemma} [Sketch decomposition] 
\label{lm:completion}
Suppose the adjacency matrix $\bA$ follows Data Model~\ref{datamodel}.
Let $\zeta = \frac{C \log^2 N}{ \: \obsprob \densdiff^2}$, where $C$ is a constant real number.  If 
\begin{align}
N' &\le \min \left\{ \nmin^2 / \zeta \:,\: N \right\}, \label{eqn:21973}
\\
N' &\ge 4 \balRS \left[ \balRS \zeta + \log \left( 2 r N \right) \right], \label{eq:lemm2_suffRS}
\end{align}
then the optimal point of \eqref{eq:convex_sketch} yields the exact low rank component of $\bA'$ with probability at least $1-c {N}^{-10}- N^{-1}$, where $c$ is a constant real number.
\end{lemma}

\begin{remark} \label{rem:RSDecomp}

The sufficient number of samples is $\softOmega\left( \frac{\balRS^2}{\obsprob \densdiff^2} \right)$.
However, if the graph is perfectly balanced, i.e., $\nmin = N/r$, then $N'=\softOmega\left( \frac{r^2}{\obsprob \densdiff^2} \right)$.
Thus, for balanced graphs, the sufficient number of randomly sampled nodes is virtually independent of the size of the graph.
The density difference $\densdiff$ and observability $\obsprob$ are also critical factors determining the number of samples.  A decrease in either parameter will drive the sufficient number of samples higher.
The sufficient condition \eqref{eqn:21973} also includes two upper bounds. First, we cannot sample more columns than there are in the full matrix.  Second, we need to avoid the following issue: by sampling a large number of columns, the supply of columns from the smallest cluster may be exhausted. Such an event will lead to \emph{more} imbalance in the sketch as sampling continues, making decomposition of the sketch even less likely to succeed.

From conditions \eqref{eqn:21973} and \eqref{eq:lemm2_suffRS}, in order for there to be a gap between the upper and lower bounds as $N \rightarrow \infty$, we need $\nmin = \softOmega\left( \frac{\sqrt{N} \log N }{\sqrt{\obsprob} \densdiff^2} \right)$.
Therefore, although URS may improve the sample and computational complexity (see Section~\ref{sec:ComputComplex} for details), the sufficient condition on cluster size remains about the same as for full-scale decomposition (in fact, because of the retrieval step, URS gives slightly poorer performance).

\end{remark}

\subsection{Full Data Clustering}
The third step infers the partition of the full graph. This is accomplished by checking the edges of each node in the full graph and finding the cluster that the edge patterns have the strongest correlation with.
\begin{algorithm}
\caption{Efficient cluster retrieval for full graph}
{\footnotesize
\textbf{Input}: Given adjacency matrix $\bA \in \mathbb{R}^{N \times N} $\\
\textbf{1. Random Node Sampling:} \label{alg:SampleStep}

\textbf{1.1} Form the set $\calI$ containing the indices of $N'$ randomly sampled nodes without replacement.
Sampling is accomplished using either URS, SbS, or SRS.
For sampling with SRS, use Algorithm~\ref{alg:SRSnodesampling} as presented in Section~\ref{sec:SRS}.

\textbf{1.2} Construct $\bA' \in \mathbb{R}^{N' \times N'}$ as the sub-matrix of $\bA$ corresponding to the sampled nodes.  

\smallbreak

\textbf{2. Subgraph Clustering:}

\textbf{2.1} Define $\bL'_*$ and $\bS'_*$
as the optimal point of
\begin{eqnarray}
\label{eq:convex_sketch}
\begin{aligned}
& \underset{\dot{\bL}',\dot{\bS}'}{\min} & & \lambda \| \dot{\bS}' \|_1  + \| \dot{\bL}' \|_* \quad\\
& \text{subject to} & &     \Omega_{\text{obs}} \left( \dot{\bS}' + \dot{\bL}' \right) = \Omega_{\text{obs}}(\bA'). \\
\end{aligned}
\end{eqnarray}
Solve optimization problem \eqref{eq:convex_sketch} with $\lambda = \frac{1}{32 \sqrt{N' \overline{\obsprob}}}$, where $\overline{\obsprob}$ is the empirical observation probability of the graph instance.

Optionally, Algorithm 1 of \cite{Chen:2014:CPO:2627435.2670322} may be used to solve optimization problem \eqref{eq:convex_sketch}, while iteratively searching for an improved value of $\lambda$ that yields better results.

\textbf{2.2} Cluster the subgraph corresponding to $\bA'$ using $\bL'_{*}$ (we use Spectral Clustering in our experiments). 

\textbf{2.3} 
If $\hat{r}$ is the number of detected clusters, define $\{ \bv_i \in \mathbb{R}^{N'} \}_{i=1}^{\hat{r}}$ as the collection of characteristic vectors of the actual clusters in the sketch matrix, i.e. the set of vectors which span the column space of $\bL'_{*}$.

\smallbreak
\textbf{3. Full Data Clustering:}\\ 
Define ${\ba_k}_{\calI} \in \mathbb{R}^{N'}$ as the vector of elements of $\ba_k$ ($k\textsuperscript{th}$ column of $\bA$) indexed by set $\calI$. Let $\hat{n}_i'$ be the number of elements in the $i\textsuperscript{th}$ cluster of the sketch (as identified in Step 2 of the algorithm). \\
\textbf{For} $k$ from 1 to $N$\\
\indent \quad $u = \arg \max_{i} \frac{ \left({\ba_k}_{\calI}\right)^T \bv_i}{\hat{n}_i'} $ \\
\indent \quad Assign the $k\textsuperscript{th}$ node to the $u\textsuperscript{th}$ cluster. \\
\textbf{End For}

}
\label{alg:mainalgorithm}
\end{algorithm}
The following lemma gives a sufficient condition for the full data clustering to succeed whp. This condition is given in terms of $\nminp$, the size of the smallest cluster in the \emph{sketch} matrix.
\begin{lemma}[Retrieval] \label{lem:Retrieve}

If
\begin{align} \label{eqn:retrieval_suff}
n_{min}' \ge \frac{ 8 p }{\obsprob \densdiff^2} \log\left(rN^2\right),
\end{align}
then step 3 (retrieval) of Algorithm~\ref{alg:mainalgorithm} will exactly reconstruct matrix $\bL$ with probability at least $1-N^{-1}$.
\end{lemma}
We can readily state Theorem \ref{thm:maintheoremRS} which establishes conditions for Algorithm~\ref{alg:mainalgorithm} to achieve exact recovery whp.
\begin{theorem}
\label{thm:maintheoremRS}
Suppose the adjacency matrix $\bA$ follows Data Model~\ref{datamodel}.
If
\begin{align}
\nmin &\ge \frac{ 8 p }{\obsprob \densdiff^2} \log\left(rN^2\right), \label{eqn:8987239}
\\
N' &\le \min \left\{ \nmin^2 / \zeta \:,\: N \right\}, \label{eqn:2348798}
\\
N' &\ge 4 f \max \left\{ f \zeta \:,\: \frac{4p}{\obsprob \densdiff^2} \log\left(rN^2\right) \right\}
\nonumber \\
&\qquad + 4f \log \left( 2 r N \right), \label{eqn:8923478789}
\end{align}
and node sampling is performed using URS without replacement, then Algorithm 1 exactly clusters the graph with probability at least $1 - c N^{-10} - 3N^{-1}$, where $\zeta$ is defined in Lemma \ref{lm:completion}.
\end{theorem}

\begin{remark} \label{rem:RSMainTheorem}

The sufficient condition is essentially that of Lemma~\ref{lm:completion}, with two additional constraints to ensure that the retrieval step is successful: condition \eqref{eqn:8987239} imposes a constraint on the minimum cluster size in the full graph, and the lower bound on $N'$ in \eqref{eqn:8923478789} is modified to ensure the retrieval step is successful. Nonetheless, with these additional constraints we still have $N'=\softOmega\left( \frac{\balRS^2}{\obsprob \densdiff^2} \right)$ as described in Remark~\ref{rem:RSDecomp}.

\end{remark}

\section{Improved Sampling Methods for Capturing Small Clusters} \label{sec:ProposedSbSSRS}

The URS method achieves significant computational complexity improvement for balanced graphs.
However, imbalance in the graphs will tend to carry over to the sketch, thus requiring a large number of samples to ensure that the sketch adequately captures the small clusters.
We now present two sampling methods which can yield a more balanced sketch of unbalanced data by sampling more from the small clusters.

\subsection{Sparsity-based Sampling} \label{sec:Sbs}

This first method, designated `Sparsity-based Sampling', samples the sparser columns of the adjacency matrix, which represent nodes with fewer connections, with a higher probability.
Specifically, the probability of the $i\textsuperscript{th}$ node being selected is set inversely proportional to the degree of the node, i.e., the $\ell_0$-norm of the corresponding column $ \| \ba_i \|_0$, a factor which is a measure of sparsity.
When calculating the norm, the entries corresponding to unobserved edges are set to zero.  Note that, because the diagonal element is always populated with one, the $\ell_0$ norm will always be greater than zero.
The sampling probabilities are properly normalized so that they sum up to one.
Here, SbS is analyzed and shown to also reduce the computational complexity, while at the same time improving the probability of success when working with unbalanced data.
The proof of the results for the SbS approach can be found in Appendix~\ref{sec:ProofsSbS}.

To clearly illustrate the advantage of this method, we will first consider a sketch produced by sampling with replacement, and later consider the more difficult scenario where sampling is performed without replacement.
The next proposition provides a result when the graph consists of disjoint cliques $\bA = \bL$, i.e., is uncorrupted.
\begin{proposition}
\label{lm:sbs_1}
Suppose $\bA\!=\!\bL$, where $\bL$ is as defined in Data Model \ref{datamodel}, and all edges are observed.
Then, using SbS with replacement,
the sampling probabilities from all clusters are equal to $1/r$.
\end{proposition}
Hence, there is an equal probability of sampling from each cluster regardless of the cluster sizes.
Even if the clusters are highly unbalanced, i.e., the largest cluster is much larger than the smallest cluster, SbS will tend to produce a sketch which is near-perfectly balanced.
In fact, we can still improve the balance of the sketch even if the graph is corrupted.
The next result will depend on the mean degree of a node in the smallest cluster, which is $\mu_{min} = \obsprob [(p-q) \, \nmin + q N]$.
We will always assume that
\begin{align} \label{eqn:14080123}
\mumin = \littleOmega( \log rN ).
\end{align}
This constraint on $\mumin$ is satisfied for the clustering problems considered herein (due to the required growth of $\nmin$ for successful clustering).
The following lemma provides a lower bound on the probability $\clustsampprob_{min}$ of sampling from the smallest cluster using SbS.
\begin{lemma}
\label{lem:SbS_nsamples}
Suppose the adjacency matrix $\bA$ follows Data Model~\ref{datamodel} and all conditions in the statement of Lemma~\ref{lm:completion} hold.
Define
\begin{align}
\eta &= \left[ 1 + \frac{q}{p} \left(\balRS-1\right) \right],
\label{eqn:23439878}
\\
\suffepsA
&= \sqrt{ \frac{ 6 \log( 2 N ) }{\mumin} }. \label{eqn:95422}
\end{align}
If $\suffepsA < 1$, and sampling is performed using SbS with replacement, then
\begin{align} \label{eqn:9024099}
 \clustsampprob_{min}
 \ge  \frac{1-\suffepsA}{1+\suffepsA} \frac{1}{ r \eta }
\end{align}
with probability at least $1-N^{-1}$.
\end{lemma}

\begin{remark} \label{rem:SbsSamp1}
The variable $\suffepsA$ is primarily dependent on the mean degree $\mumin$ of the smallest cluster.  This mean value may be small in some challenging cases: when few entries are observed, when $(p-q)$ is small, or when the clusters are extremely small. However, as this mean increases, the bounds on the probability improve. Due to the assumption \eqref{eqn:14080123}, we have $\frac{1-\suffepsA}{1+\suffepsA} \rightarrow 1$.

In \eqref{eqn:23439878}, the variable $\eta$ reflects an important tradeoff.
Let us first consider how $\eta$ behaves for fixed $N$ and $\nmin$.
At one extreme, as $q/p \rightarrow 1$, we have $\eta \rightarrow f$, which means that the probability will strongly depend on the proportion of the minimum cluster size to the full graph, as in URS.  At the other extreme, as $q \rightarrow 0$, we have $\eta \rightarrow 1$.  In this case, assuming $\suffepsA$ is small, the sampling probability for each cluster approaches $r^{-1}$, leading to a roughly equal chance of sampling from each cluster.
In terms of asymptotic behavior, we have $\eta = \bigO \left( q \balRS \right)$ if $q > 0$, or $\eta=1$ exactly if $q=0$.

\end{remark}

\begin{remark} \label{rem:SbsSamp2}
The $r$ in the denominator of \eqref{eqn:9024099} is due to some conservatism in the bounding techniques. However, in the regime where $\densdiff$ is small, a better approximation dispenses with this factor in the sufficient condition, such that SbS sampling probabilities become roughly equal to those of random sampling. See Remark~\ref{rem:SbsSamp3} in Appendix~\ref{sec:ProofsSbS} for more details.
\end{remark}

Now, we consider the SbS sampling process without replacement, as is used in Step 1 (sampling) of Algorithm~\ref{alg:mainalgorithm}.
If a column is sampled a second time, the duplicate sample is discarded and not counted towards the sampling budget $N'$.
The next lemma shows that as a consequence of the balanced sampling probability, fewer samples are required to obtain a balanced sketch.
\begin{lemma}
[Sampling Size for SbS] \label{lem:Suff_Optimality_Partial_SbS}

Suppose that $\bA'$ is produced by SbS with $N'$ sampled columns (without replacement).
Given an integer $\nminbound$ with $\nmin > \nminbound > 0$, let
\begin{align}
g &= \eta \frac{1+\suffepsB}{1-\suffepsB} \left( 1-\frac{\nminbound}{\nmin} \right)^{-1}, \label{eqn:45523235}
\\
\suffepsB &= \sqrt{ \frac{ 6 \log( 2 r N ) }{\mumin} }, \label{eqn:43453}
\end{align}
where $\eta$ is as defined in \eqref{eqn:23439878}.
Then, $\nminp \ge \nminbound$ with probability at least $1-N^{-1}$ provided that $\suffepsB < 1$ and
\begin{align} \label{eqn:Sbs_samples_suff}
N
\ge N'
\ge r g \left[ \nminbound \log(N) + \log\left(2rN\right) \right].
\end{align}

\end{lemma}

\begin{remark} \label{rem:Suff_Optimality_Partial_SbS}

Note that the sufficient condition of Lemma~\ref{lem:Suff_Optimality_Partial_SbS} has a similar structure to that of Lemma~\ref{lem:Suff_Optimality_Partial}.
In the sufficient condition for URS \eqref{eqn:Suff_Optimality_Partial_RS_uniform}, the main factor is the variable $f$, whereas for SbS in \eqref{eqn:Sbs_samples_suff}, this factor changes to $r g$. 

As columns are sampled from a particular cluster, the probability of sampling from this cluster may decrease. This possibility is reflected by the $\left( 1-\frac{\nminbound}{\nmin} \right)^{-1}$ penalty term which appears in \eqref{eqn:45523235}. This effect will be small for large clusters, but may become significant for small clusters as the required number of samples approaches the size of the smallest cluster.
\end{remark}

\begin{remark}

Due to assumption \eqref{eqn:14080123}, we have $\frac{1-\suffepsB}{1+\suffepsB} \rightarrow 1$.
Furthermore, suppose that $\nminbound$ does not approach $\nmin$ (for example, assume that $\nminbound \le c \nmin$ for some constant $c$).
Given these conditions, the number of samples to guarantee sufficient representation of the smallest cluster in the sketch whp is $N' = \softOmega \left( q \nminbound r \balRS \right)$ for $q>0$.
However, $q$ may scale with $N$. In particular, if $q = \bigO \left( \balRS^{-1} \right)$, then we get a much better result: $N' = \softOmega \left( \nminbound r \right)$.
We will refer to this as the \smallq/ regime. The regime where $q = \littleOmega \left( \balRS^{-1} \right)$ will be referred to as the \largeq/ regime.

\end{remark}

We now proceed to state the main result, which will consist of three theorems. First, we provide a sufficient condition on the number of samples to obtain sketch clusters of size $\bigOmega(\sqrt{N'} \log N')$, thus setting the stage for successful clustering of the sketch.
\begin{theorem}[Cluster size for SbS] \label{thm:sketchnmin}
Suppose $\bA'$ is produced by SbS with density difference $\densdiff'$ and observation probability $\obsprob'$, and let
\begin{align}
\balSbS' &= 2 \eta \frac{1+\suffepsB}{1-\suffepsB},
\\
\SbsNminpConst &= \frac{C \log^2 N}{\obsprob' {\densdiff'}^2}.
\end{align}
If
\begin{align}
N' &\le \min \left\{ \frac{\nmin^2}{4 \SbsNminpConst} \:,\: N \right\}, \label{eqn:786867123}
\\
N' &\ge  r \balSbS' \left[ r \balSbS' \SbsNminpConst \log^2 N + 2 \log(2rN) \right], \label{eq:lemm2_suffSbSA}
\end{align}
then
\begin{align} \label{eqn:9893784}
\nminp \ge \frac{\sqrt{C N'} \log N'}{\sqrt{\obsprob'} \densdiff'}
\end{align}
 with probability at least $1-N^{-1}$.

\end{theorem}

Next, we look at the stochastic properties of the sketch induced by the randomness from the probabilistic sampling procedure in addition to that of the generative SBM.
Let $p_i'$ be the intra-cluster edge density for cluster $i$ in the sketch matrix.
Likewise, let $q'$ be the inter-cluster edge density, and $\obsprob'$ the observation probability for the sketch matrix.
If we perform URS, then we have $p_i'=p$ (for all clusters), $q'=q$, and $\obsprob'=\obsprob$.
However, when the graph is sketched using SbS, there could be a bias of $p'$, $q'$, and $\obsprob'$ toward slightly smaller values than in the original graph.
The next theorem demonstrates not only that the bias will be extremely small, but also that the probabilities of the sketch asymptotically approach those of the full graph as $N \rightarrow \infty$.
\begin{theorem}[Sketch probabilities $p, q, \obsprob$ for SbS] \label{thm:sketchpqrho}

If $\bA$ is constructed using Data Model~\ref{datamodel}, and subsequently sampled using SbS, then the sketch matrix densities will be bounded by
\begin{align}
&p \ge p_i' \ge p^-,
\\
&q \ge q' \ge q^-, \label{eqn:3242334}
\\
&\obsprob \ge \obsprob' \ge \obsprob^-, \label{eqn:423423}
\end{align}
where
\begin{align}
p^- &= p ( 1 - \pqprimeeps )^2, \label{eqn:pminus}
\\
q^- &= q ( 1 - \pqprimeeps )^2, \label{eqn:qminus}
\\
\obsprob^- &= \obsprob ( 1 - \rhoprimeeps )^2, \label{eqn:obsprobminus}
\\
\pqprimeeps &= N (1-\obsprob p)^{\nmin} (1-\obsprob q)^{\nmin},
\\
\rhoprimeeps &= N (1-\obsprob)^N.
\end{align}

\end{theorem}
We can readily state the following theorem which provides guarantees for clustering using SbS based on the statistical properties from Theorem~\ref{thm:sketchpqrho}.

\begin{theorem}
\label{thm:maintheoremSbS}

Define $p'$ as the intra-cluster edge probability of the sketch (assuming the probability is the same for all clusters), and $\densdiff'=1-2\max\{1-p',q'\}$.
Suppose that
\begin{enumerate}

\item
A sampling algorithm produces a sketch following the SBM with parameters
\begin{align}
p' &= p^-,
\\
q' &= q^-,
\\
\obsprob' &= \obsprob^-,
\\
\nminp &\ge \frac{\sqrt{C N'} \log N'}{\sqrt{\obsprob'} \densdiff'}, \label{eqn:34879800}
\end{align}
where $p^-$,$q^-$, and $\obsprob^-$ are as in \eqref{eqn:pminus}-\eqref{eqn:obsprobminus}.

\item
The following two conditions hold:
\begin{align}
&\nmin \ge \frac{ 16 p }{\obsprob \densdiff^2} \log\left(rN^2\right), \label{eqn:421159}
\\
&N \ge N' \ge r \balSbS' \log(2 r N) \left[ \frac{ 16 p }{\obsprob \densdiff^2} \log(N)
+ 1 \right]. \label{eq:lemm2_suffSbSB}
\end{align}

\end{enumerate}
Then, Algorithm~\ref{alg:mainalgorithm} exactly clusters the graph with probability at least $1 - c N^{-10} - 2N^{-1}$, where $c$ is a constant real number.

\end{theorem}

\begin{remark} \label{rem:SbsMainTheorem}
Combining Theorems~\ref{thm:sketchnmin}-\ref{thm:maintheoremSbS}, the sampling complexity will be roughly $\softOmega\left( \frac{r^2 q^2 \balRS^2}{\obsprob \densdiff^2} \right)$.
The main savings in this sampling complexity come when $q$ is sufficiently small.  In the \smallq/ setting, the sampling complexity becomes $\softOmega\left( \frac{r^2}{\obsprob \densdiff^2} \right)$, thus making the sufficient condition almost independent of the graph size.
In the \largeq/ regime, we need $\nmin = \softOmega \left( \frac{ r \sqrt{q N} }{\sqrt{\obsprob} \densdiff^2} \right)$,
while in the \smallq/ regime, we only need $\nmin = \softOmega \left( \frac{ r }{\obsprob \densdiff^2} \right)$. In the \smallq/ regime, the number of clusters must be $\bigO \left( \frac{\densdiff \sqrt{\obsprob N}}{\log^2 N} \right)$.

As shown in Table~\ref{tab:AsymptoticComparison}, the best achievability results for a convex algorithm are $\nmin=\bigOmega \left( \frac{\log N}{\obsprob \densdiff^2} \vee \frac{\sqrt{q N}}{\sqrt{\obsprob}\densdiff} \right)$ for the \largeq/ regime and $\nmin=\bigOmega \left( \frac{\log N}{\obsprob \densdiff^2} \right)$ for the \smallq/ regime \cite{Chen:2016:STP:2946645.2946672,cai2015}. In addition to the aforementioned computational gain of our sketch-based approach, our result compares favorably with these works considering that the clustering algorithm we use only guarantees success if $\nmin = \bigOmega \left( \frac{ \sqrt{N} \log N }{\sqrt{\obsprob} \densdiff} \right)$ when clustering the full graph.
Additionally, the analysis of \cite{Chen:2016:STP:2946645.2946672} requires that clusters be of equal size, and both \cite{Chen:2016:STP:2946645.2946672} and \cite{NIPS2016_6574} require strong side information, i.e., the sum of squares of cluster sizes.
\end{remark}

Note that SbS only looks at the degree, which is a count of edges. One would expect that taking into account the individual edges making up this count would yield further improvement. Next, we present a method which does just this.

\subsection{Spatial Random Sampling} \label{sec:SRS}

In Section~\ref{sec:Sbs}, we presented 
SbS which  notably raises the chance of sampling from small clusters as compared to URS.
In this section, we present our second node sampling method with which we can obtain a balanced sketch from unbalanced data. In \cite{7968311}, two of the authors proposed Spatially Random Sampling (SRS) as a new randomized data sampling method. The main idea underlying SRS is to perform the random sampling in the spatial domain. To this end, the data points (here, the columns of the adjacency matrix) are projected on the unit sphere, then points are sampled successively based on their proximity to randomly chosen directions in the ambient space. Thus, with SRS the probability of sampling from a specific data cluster depends on the amount of space the cluster occupies on the unit sphere (since SRS is applied to the normalized unit $\ell_2$-norm data points). Accordingly, the probability of sampling using SRS is nearly independent of the population sizes of the clusters \cite{7968311}.

Suppose $\bA = \bL$. In this case, the columns of $\bA$ lie in the union of $r$ 1-dimensional subspaces. In fact, these subspaces are orthogonal since the clusters are non-overlapping. The following lemma shows that if we use SRS for node sampling in this case, the probabilities of sampling from the clusters will be equal regardless of how unbalanced the data is.
For further details, we refer the reader to \cite{7968311}.
\begin{lemma}
Suppose $\bA = \bL$ and SRS is used to sample one column of $\bA$. If the corresponding node is sampled, then the probability of sampling from each cluster is equal to $1/r$. 
\end{lemma}

Algorithm \ref{alg:SRSnodesampling} contains the SRS node sampling algorithm.
We now describe the main steps of this algorithm.
\begin{algorithm}
\caption{Efficient Node Sampling using SRS}
{\footnotesize
\textbf{Input}: Given adjacency matrix $\bA \in \mathbb{R}^{N \times N} $
\smallbreak

\textbf{1. Pre-completion}: Apply Algorithm \ref{alg:precomplete} to $\bA$ to obtain $\bA^c$.

\textbf{2. Random embedding:} Calculate $\bA^c_{\phi} = \mathbf{\Phi} \bA^c$, where $ \mathbf{\Phi} \in \mathbb{R}^{m \times N}$ is a random binary matrix. Normalize the $\ell_2$-norm of the columns of $\bA^c_{\phi}$, i.e., set ${\ba^c_{\phi}}_i = \frac{{\ba^c_{\phi}}_i}{ \| {\ba^c_{\phi}}_i \|_2}$, where ${\ba^c_{\phi}}_i$ is the $i^{th}$ column of $\bA^c_{\phi}$.

\textbf{3. Node Sampling:} Invoke Algorithm~\ref{alg:SRSColumnSampling} to sample $N'$ nodes/columns (without replacement) from matrix $\bA^c_{\phi}$. 

\smallbreak
\textbf{Output:} Set $\calI$ as the set of the indexes of sampled nodes/columns in Step 3. 
}
\label{alg:SRSnodesampling}
\end{algorithm}

\subsubsection{Preparing data for SRS}
Because SRS is not designed to support missing values, an important issue with adopting SRS as our node sampling algorithm is the missing values of the adjacency matrix $\bA$. One possible and easy solution is to replace the missing values with zeros. While this works well when a small fraction of the elements are missing, it degrades the performance of SRS if a notable part of the adjacency matrix is missing. The main reason stems from the previous observation, that is, if the data is clean and complete, then the columns corresponding to a cluster lie in a 1-dimensional subspace. If we replace the missing values with zeros, the columns corresponding to the large clusters significantly diffuse in the space. Thus, SRS ends up sampling many columns from the large clusters. We address this problem by pre-completing the adjacency matrix. The larger clusters are easily captured using random sampling. Thus, we apply a modified version of Algorithm~\ref{alg:mainalgorithm} with URS to complete $\bA$.
Algorithm \ref{alg:precomplete} presents the pre-completion step.

\begin{algorithm}
\caption{Data Pre-Completion}
\label{alg:precomplete}
{\footnotesize
\textbf{Input}: Given adjacency matrix $\bA \in \mathbb{R}^{N \times N} $
\smallbreak

\textbf{1. Random Node Sampling:}

\textbf{1.1} Form the set $\calI$ consisting of indices of $N'$ randomly sampled nodes and construct $\bA' \in \mathbb{R}^{N' \times N'}$ as the sub-matrix of $\bA$ corresponding to the sampled nodes.  

\smallbreak

\textbf{2. Subgraph Clustering (same as step 2 of Algorithm~\ref{alg:mainalgorithm}):}

\textbf{2.1}
Solve optimization problem \eqref{eq:convex_sketch} with $\lambda = \frac{1}{32 \sqrt{N' \overline{\obsprob}}}$ to get optimal points $\bL'_{*}$ and $\bS'_{*}$.
Optionally, Algorithm 1 of \cite{Chen:2014:CPO:2627435.2670322} may be used.

\textbf{2.2} Cluster the subgraph corresponding to $\bA'$ using $\bL'_{*}$ (we use Spectral Clustering in our experiments). 

\textbf{2.3} 
If $\hat{r}$ is the number of detected clusters, define $\{ \bv_i \in \mathbb{R}^{N'} \}_{i=1}^{\hat{r}}$ as the collection of characteristic vectors of the actual clusters in the sketch matrix.

\smallbreak
\textbf{3. Adjacency Matrix Generation:}\\ 
\textbf{3.1} Initialize $\bU \in \mathbb{R}^{N \times \hat{r}}$ as a zero matrix.\\
\textbf{3.2 For} $k$ from 1 to $N$\\
\indent \quad $j = \arg \min_{i} \| {\ba_k}_{\calI} - \bv_i \|_2  $ \\
\indent \quad If $j > 0$, then $u_{kj} = 1$, where $u_{kj}$ is the element of matrix $\bU$ at the $k\textsuperscript{th}$ row and $j\textsuperscript{th}$ column.
\\
\indent \quad Note: The index of the columns/rows starts from 1 (similar to MATLAB).

\textbf{3.2 End For}

\smallbreak
\textbf{4. Output:} Compute the completed matrix $\bA^c = \bU \bU^T + \bA$ and clamp the elements of $\bA^c$ which are greater than 1 to 1.

}
\end{algorithm}

\subsubsection{Random Embedding}
In contrast to the low rank approximation-based column sampling methods, SRS is not sensitive to the existence of linear dependence between the clusters \cite{7968311}. Thus, in order to reduce the computational complexity, first we embed the columns of the adjacency matrix using a computationally efficient embedding method to reduce the dimensionality of $\bA$. For instance, we can use a random binary embedding matrix (whose elements are independent random variables with values $\pm 1$ with equal probability) to embed the 
columns of $\bA$ into a lower dimensional space and apply the SRS algorithm to the embedded data. Data embedding using a random binary matrix is significantly faster than using a conventional random Gaussian matrix since it does not involve numerical multiplication, and comes with no significant loss in performance.

\subsubsection{Column Sampling}
Finally, we invoke the SRS column sampling algorithm of \cite{7968311}, which is listed here in Algorithm~\ref{alg:SRSColumnSampling}.

\begin{algorithm}
\caption{Sample $N'$ columns of $\bD  \in \mathbb{R}^{m \times N}$ using Spatial Random Sampling (SRS) \cite{7968311}}
{\footnotesize
\textbf{Input}: Data matrix $\bD$ and $N'$ as the number of sampled columns.

\smallbreak
\textbf{Initialization}: Construct matrix $\mathbf{\Phi} \in \mathbb{R}^{N' \times m}$ by sampling independently from $\calN (0,1)$. Set $\bY$ equal to an empty matrix.

\smallbreak
\textbf{1. Data Normalization}:
Define $\bX \in \mathbb{R}^{m \times N}$ such that $\bx_i = \bd_i / \| \bd_i \|_2$.

\smallbreak
\textbf{2. Column Sampling}:\\
\textbf{2.1} Set $\bQ =    \mathbf{\Phi} \bX  $ and set $\calI = \emptyset$.

\textbf{2.2 For} $i = 1$ to $N'$\\
\textbf{2.2.1} Define $\bh = \left| \bq^i \right|$ and set $\bh_{\calI} = 0$ where $\bq^{i}$ is the $i^{\text{th}}$ row of $\bQ$ and $\bh_{\calI}$ are the elements of $\bh$ with indices in $\calI$. \\
\textbf{2.2.2} Append $k$ to $\calI$ where $k$ is the index of the maximum element of $\bh$. \\
\textbf{2.2 End For}

\smallbreak
\textbf{Output:}  $\calI$ as the set of the indices of the sampled columns. 
}
\label{alg:SRSColumnSampling}
\end{algorithm}

\section{Computational Complexity} \label{sec:ComputComplex}

We first consider the complexity of a single iteration of the convex optimization in Step~2 of Algorithm~\ref{alg:mainalgorithm}, which comprises the principal cost in most cases.
The complexity of each iteration is dominated by the SVD computation, requiring $\bigOmega(r N^2)$ computations per iteration for the full graph, and $\bigOmega(r N'^2)$ for the sketch.

If URS is used, based on the required number of samples, we have a cost of order $\softO \left( \frac{ r f^4 }{\obsprob^2 \densdiff^4} \right)$ per iteration.  If the graph is perfectly balanced, then the cost reduces to only $\softO \left( \frac{ r^5 }{\obsprob^2 \densdiff^4} \right)$, a considerable saving when the number of clusters scales sublinearly with $N$.  Likewise, the sketch subgraph can be as small as $\softO(r^2)$, which can significantly reduce the memory requirements, or even allow processing of large graphs that are virtually impossible to cluster otherwise.

For SbS, the computational complexity is roughly $\softO\left( \frac{r^3 q^4 f^4}{\obsprob^2 \densdiff^4} \right)$ for the \largeq/ regime.
However, in the \smallq/ regime, this can become as small as $\softO\left( \frac{r^5}{\obsprob^2 \densdiff^4} \right)$, making the decomposition almost independent of the graph size.

The computational complexity of Step~\ref{alg:SampleStep} of Algorithm~\ref{alg:mainalgorithm} will depend on the sampling algorithm used. For URS, this is linear in $N'$ and independent of the graph size. For SbS, a linear number of $\ell_0$-norms must be calculated. For SRS, the sampling step is of order $\bigO(N' N^2)$, i.e., for each sample, we need an inner product with each column, albeit the binary embedding makes these inner products cheaper, of order $\bigO(m N' N)$, where $m$ is the dimension of the embedding space. 

The complexity of the retrieval step is $\bigO(r N' N)$ since it requires an inner product with each column for each cluster.

\section{Numerical and Experimental Results}
\label{sec:numericalresults}
In this section, we evaluate the performance of the proposed randomized framework and compare it to algorithms which cluster the full-scale graph. First, we demonstrate the substantial speedups afforded by the proposed approach in balanced graphs, and then perform experiments with unbalanced data to showcase the effectiveness of the SbS and SRS schemes. Cases are presented in which the proposed method with SbS and SRS can even outperform full-scale decomposition. An experiment is considered successful if the algorithm exactly reconstructs the low rank matrix $\bL$ with no errors.
Unless otherwise noted, each data point in the plots is obtained by averaging over 20 independent runs.

After solving the optimization problem of \cite{Chen:2014:CPO:2627435.2670322,cai2015}, we estimate the clusters using Spectral Clustering \cite{6482137} applied to the obtained low rank component.
As with \cite{NIPS2014_5309}, we use $\lambda=1 / \sqrt{N'}$ (or $\lambda=1 / \sqrt{N}$ when clustering the full graph), a value which we found to work well across a wide range of regimes.
In all line plots, we compare against the full graph clustering algorithms \cite{cai2015,NIPS2016_6574,Chen:2014:CPO:2627435.2670322}.

\subsection{Clustering Balanced Graphs}
We first consider the case of balanced graphs, where all three of the sampling methods yield speedups without sacrificing accuracy.
The parameters of the Data Model~\ref{datamodel} are set to $r = 2$, $p=0.8$, $q = 0.1$,  $\obsprob = 0.7$, and $n_1 \!=\! n_2 \!=\! N/2$.
For the randomized method we use URS with $N' = 200$ samples.
The running time in seconds is shown as a function of $N$ in Fig. \ref{fig:Balanced}(a).
The full-scale and sketching-based approaches cluster the data accurately in every case. However, in all cases the randomized method is substantially faster, running in less than 2.5 seconds. The full-scale decomposition \cite{Chen:2014:CPO:2627435.2670322}, on the other hand, ranges from 4.1 seconds for $N=500$ to 661.5 seconds for $N=10000$, with the other full-scale algorithms running slower than this.
The main factor in the fast execution speed of the randomized approach is that the low rank plus sparse decomposition algorithm is applied to the much smaller sketch matrix. While the full-scale clustering algorithms exhibit a polynomial increase in runtime with respect to $N$, because the sketch matrix is fixed in size for this simulation, the URS runtimes are roughly constant as $N$ increases (SbS and SRS are not shown due to their similar timings in this regime).
The complexity of the retrieval step of Algorithm~\ref{alg:mainalgorithm} is linear with $N$, making its run time impact insignificant.

Fig.~\ref{fig:Balanced}(b) shows the phase transition of the randomized approach with URS in terms of $N$ and $N'$, using the same scenario as in Fig.~\ref{fig:Balanced}(a). Even with only 75 samples, the proposed approach can achieve 100\% success for $500 \le N \le 5000$.

\begin{figure}
\centering
\includegraphics[scale=1]{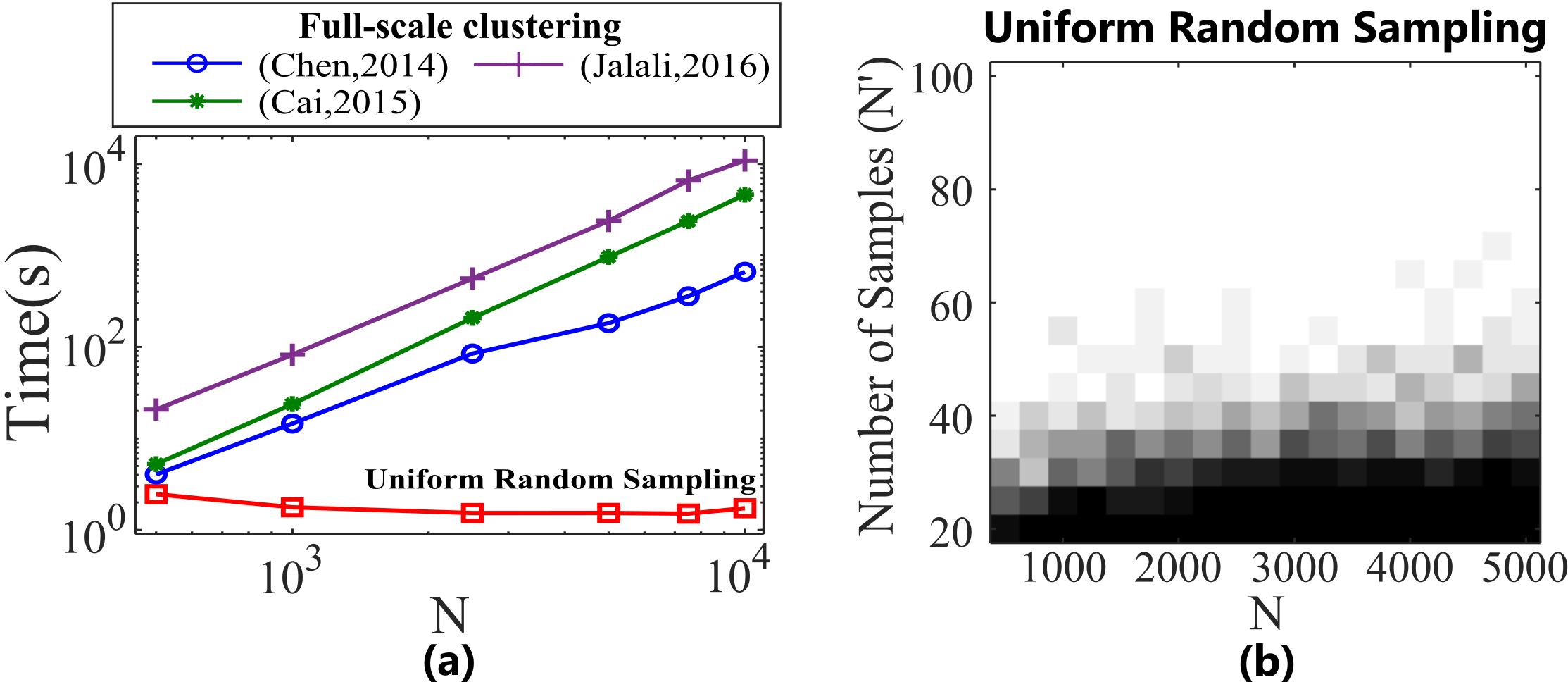}
\vspace{-.65cm}
\caption{(a) Timing comparison between Algorithm~\ref{alg:mainalgorithm} with URS, full-scale decomposition with (Chen,2014) \cite{Chen:2014:CPO:2627435.2670322}, (Cai,2015) \cite{cai2015}, and (Jalali,2016) \cite{NIPS2016_6574}.
Times are averaged over five runs.
(b) Phase transition plot for URS.  White regions indicate success and black regions failure.
}
\label{fig:Balanced}
\end{figure}

\subsection{Clustering Unbalanced Graphs}

\subsubsection{Study of proposed sampling methods}

In this experiment, we demonstrate the performance of the proposed randomized approach in terms of sampling complexity and minimum cluster size.
The graph follows Data Model~\ref{datamodel} with parameters $p=0.8$, $q=0.1$,  $\obsprob=0.7$, and $N=5000$. The graph consists of two small clusters with sizes $n_1 \!=\! n_2 \!=\! \nmin$, and one large cluster with size $n_3 \!=\! 5000 \!-\! 2 \nmin$.
Phase transition plots are shown in Fig.~\ref{fig:RandPhasePlots}, comparing the URS, SbS, and SRS sampling techniques with respect to sample complexity and minimum cluster size.
For SRS, we set $m=500$ in Algorithms~\ref{alg:SRSnodesampling} and \ref{alg:SRSColumnSampling}. Additionally, we found that sampling part of the sketch using URS improved the success rate: here we acquire $N'/2$ samples via URS, and $N'/2$ via SRS.

The phase transitions are shown over the domain $50 \!\le\! N' \!\le\! 700$ and $50 \!\le\! \nmin \!\le\! 700$. The algorithm which uses SbS can yield exact clustering even when the URS-based algorithm fails due to highly unbalanced data. For example, when $\nmin \!=\! 200$, the SbS algorithm extracts accurate clusters using only $N' \!=\! 200$ samples (only 4\% of the total nodes).
\begin{figure}
\centering
\includegraphics[scale=1]{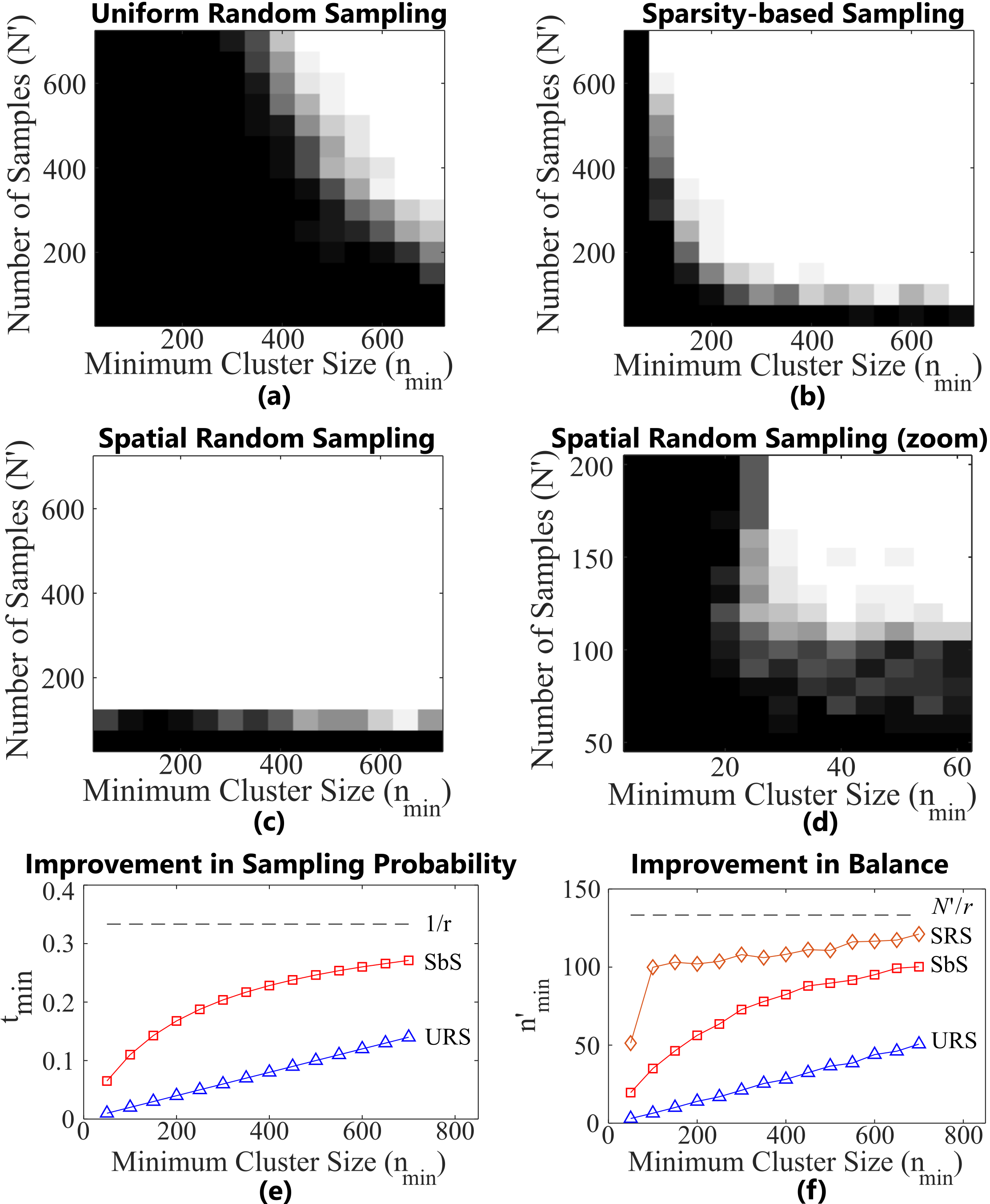}
\vspace{-.65cm}
\caption{Phase transition plots for (a) URS, (b) Sparsity-based Sampling, and (c) Spatial Random Sampling.
(d) provides a zoom into the lower left-hand region of (c).
For fixed $N'=400$. (e) shows the probability of sampling from the smallest cluster, and (f) shows the minimum size of the sketch clusters.
}
\label{fig:RandPhasePlots}
\end{figure}
SRS performs exceedingly well for very small cluster sizes, so much so that that the phase transition cannot be seen in (c).
Fig. ~\ref{fig:RandPhasePlots}(d) shows a zoom of the lower-left portion of panel (c) so that the phase transition is visible.

The improved performance of SbS is due to the more balanced sketches, stemming from the larger sampling probability that SbS places on the columns in the small clusters.
Fig.~\ref{fig:RandPhasePlots}(e) shows the probability with which the smallest cluster is sampled for URS and SbS, corresponding to the $N'=400$ row in Fig.~\ref{fig:RandPhasePlots}(a) and (b). The dashed line shows the ideal probability in which each cluster has equal sampling probability.
The probability for URS is approximately $\nmin / N$, while SbS provides at least two-fold improvement over URS.

Fig.~\ref{fig:RandPhasePlots}(f) shows the resulting minimum cluster sizes in the sketch, again for $N'=400$.
The black dashed line shows the ideal minimum cluster size to attain the most balanced sketch, i.e. where the clusters are equal-sized.
For URS and SbS, we see the same trends as found in Fig.~\ref{fig:RandPhasePlots}(e).
We also show the results for SRS, where the proportions come close to the ideal.

\subsubsection{Comparison of proposed algorithm with full scale clustering}

In addition to being significantly faster than the full-scale clustering algorithms, the proposed algorithm can even \textit{outperform} the full-scale algorithm in terms of success rate.
As described in Remark~\ref{rem:SbsMainTheorem}, this occurs when the inter-cluster probability $q$ is sufficiently small.
We remark that this does not violate the data processing inequality which indicates that post-processing cannot increase information \cite{Cover:2006:EIT:1146355}.
Rather, many full-scale clustering algorithms are not robust to data unbalancedness in the sense that they often fail to yield accurate clustering with unbalanced data. 

As an example, consider the scenario where $p=0.6$, $q=0.01$, $\obsprob=0.4$, $N=5000$, with a graph composed of three clusters with two small clusters of size $n_1 \!=\! n_2 \!=\! \nmin$ and one dominant cluster of size $n_3 \!=\! 5000 \!-\! 2 n_{min}$. Fig.~\ref{fig:FullVsSamp}(a) compares the probability of success of the proposed randomized with SbS and the three full-scale clustering algorithms, as a function of $\nmin$.  Using SbS and SRS with 800 sampled nodes (only 16\% of the total number nodes), the proposed approach yields exact clustering even when $\nmin = 120$. On the other hand, when $\nmin \le 220$, the full-scale decomposition algorithm \cite{Chen:2014:CPO:2627435.2670322} fails to yield accurate clustering.
\begin{figure}
\centering
\includegraphics[scale=1]{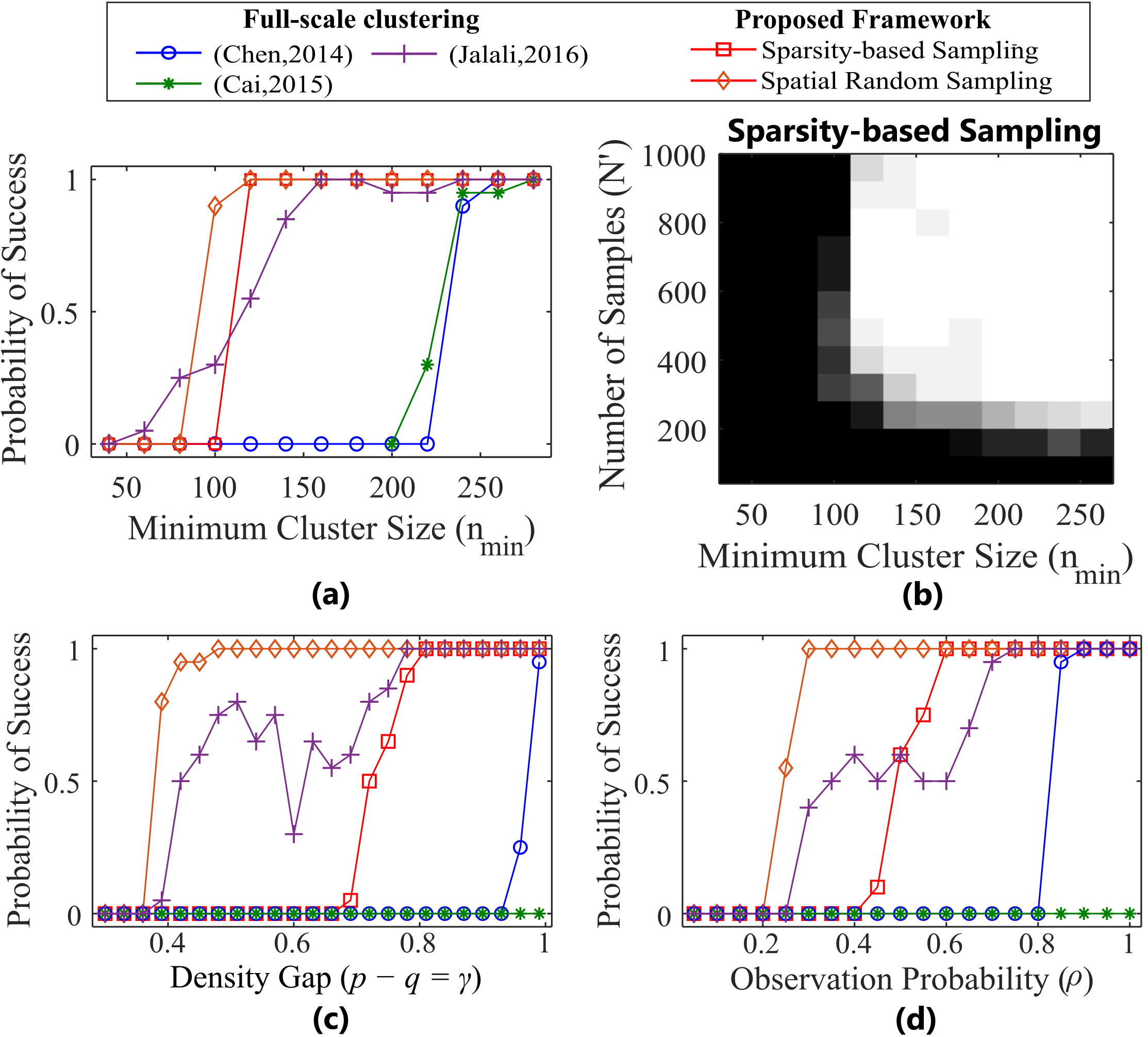}
\vspace{-.65cm}
\caption{(a) Success probability for full-scale clustering and SbS approaches as a function of $\nmin$.  (b) Phase transition with SbS.}
\label{fig:FullVsSamp}
\end{figure}
While \cite{cai2015} can have stronger asymptotic guarantees than \cite{Chen:2014:CPO:2627435.2670322} if $q \rightarrow 0$, in this case it has similar performance.
In fact, the \largeq/ lower bound on $\nmin$ for \cite{cai2015} (see Table~\ref{tab:AsymptoticComparison}) consists of a maximum over two terms. For these parameters, the two terms are almost equal, suggesting that $q$ is large enough to degrade the performance of \cite{cai2015}.
The algorithm of \cite{NIPS2016_6574} experiences successes for small $\nmin$, but does not have a sharp phase transition, and underperforms SbS and SRS in terms of exact recovery.
We note that URS fails in all regimes of Fig.~\ref{fig:FullVsSamp} and Fig.~\ref{fig:Varyingr}, and we therefore omit it from the plots.

Fig.~\ref{fig:FullVsSamp}(b) shows the phase transition in terms of $\nmin$ and $N'$ of the randomized approach with SbS under the same setup as in Fig.~\ref{fig:FullVsSamp}(a). This plot shows that SbS has the potential to perform well even if we reduce the number of samples.

In Fig.~\ref{fig:FullVsSamp}(c), a similar plot is shown as in Fig.~\ref{fig:FullVsSamp}(a), but with varying density difference $\densdiff$. We set $p\!=\!(1\!+\!\densdiff)/2$ and $q\!=\!(1\!-\!\densdiff)/2$, which is equivalent to setting the density difference and density gap equal, i.e. $\densdiff \!=\! p\!-\!q$. For this plot, the remaining parameters are $\nmin \!=\! n_1 \!=\! n_2 \!=\! 90$, $n_3 \!=\! 4820$, $\obsprob \!=\! 0.7$, $N' \!=\! 500$.
Likewise, in Fig.~\ref{fig:FullVsSamp}(d), we show results for varying observation probability $\obsprob$, where $\nmin \!=\! n_1 \!=\! n_2 \!=\! 100$, $n_3 \!=\! 4800$, $p \!=\! 0.8$, $q \!=\! 0.1$, $N' \!=\! 1000$.

In all of the previous examples, we set $r\!=\!3$. By having two small clusters of equal size, we prevent the possibility that the clusters are identifiable purely by observing the degree (this may be possible, for example, if $r\!=\!2$ with a very large and a very small cluster).
In Fig.~\ref{fig:Varyingr}, we show results when $r$ is increased with clusters of varying size.
\begin{figure}
\centering
\includegraphics[scale=1]{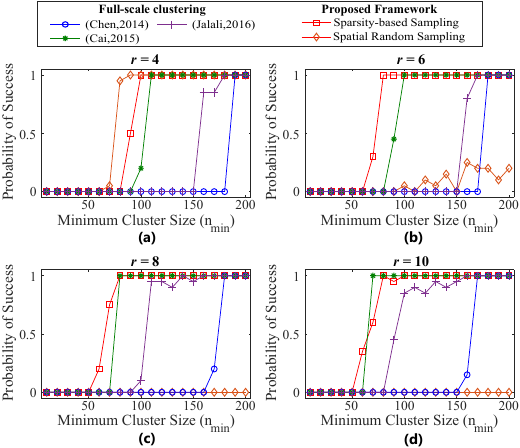}
\vspace{-.65cm}
\caption{Phase transitions for different numbers of clusters.}
\label{fig:Varyingr}
\end{figure}
In all cases, the following parameters are used: $N \!=\! 6000$, $N' \!=\! 1000$, $\obsprob \!=\! 0.4$, $p \!=\! 0.6$, $q \!=\! 0.01$.
For Fig.~\ref{fig:Varyingr}(a), there are four clusters with cluster sizes $n_1\!=\!n_2 \!=\! \nmin$, and $n_3 \!=\! n_4 \!=\! (N-2\nmin)/2$.
For Fig.~\ref{fig:Varyingr}(b), there are six clusters with cluster sizes $n_1 \!=\! n_2 \!=\! \nmin$, $n_3 \!=\! n_4  \!=\! 200$, and $n_5 \!=\! n_6 = (N-400-2\nmin)/2$.
For Fig.~\ref{fig:Varyingr}(c), there are eight clusters with cluster sizes $n_1 \!=\! n_2 \!=\! n_3 \!=\! \nmin$, $n_4 \!=\! n_5 \!=\! n_6 \!=\! 200$, and $n_7 \!=\! n_8 \!=\! (N-600-3\nmin)/2$.
Finally, for Fig.~\ref{fig:Varyingr}(d), there are 10 clusters with cluster sizes $n_1 \!=\! n_2 \!=\! n_3 \!=\! n_4 \!=\! \nmin$, $n_5 \!=\! n_6 \!=\! 300$, $n_7 \!=\! n_8 \!=\! 200$, and $n_9 \!=\! n_{10} \!=\! (N-1000-4\nmin)/2$.
These results show that SbS performs competitively, even against \cite{cai2015}, for all the presented graph structures, but the performance of SRS degrades in these scenarios as $r$ increases.

\subsubsection{Performance of SRS with many inter-cluster errors}

Now, we turn our focus to SRS.
Smaller values of $p$ and $\rho$ will increase the dispersion of the larger cluster, thus degrading the results of SRS. However, this degradation is offset when the value of $q$ is increased, thus causing dispersion of the smaller clusters as well.
To illustrate this, in Fig.~\ref{fig:SBSvsSRShighq}, we show the results when $q$ is larger.  As in Fig.~\ref{fig:RandPhasePlots}, we set $p=0.8$,  $\obsprob=0.7$, and $N=5000$, but here $q=0.23$.  
\begin{figure}
\centering
\includegraphics[scale=1]{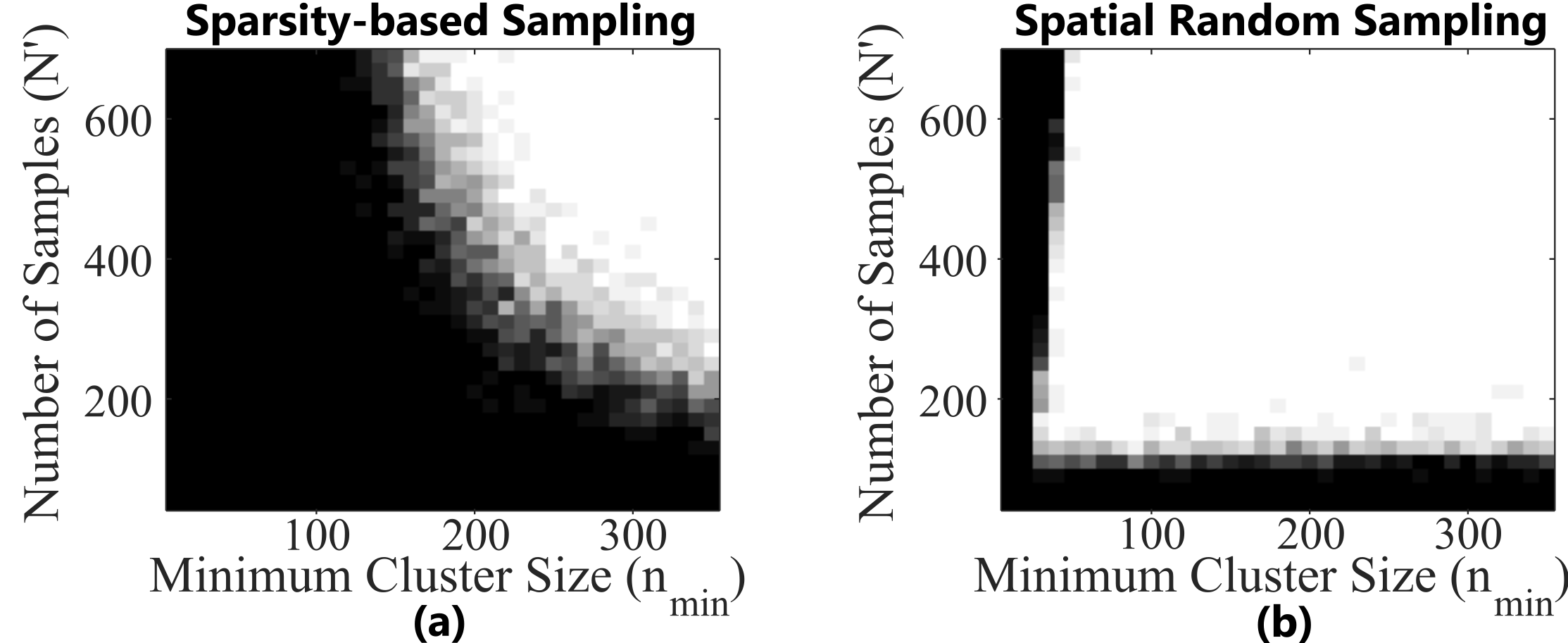}
\vspace{-.65cm}
\caption{Comparison of (a) SbS and (b) SRS for $q$ large.}
\label{fig:SBSvsSRShighq}
\end{figure}
In this regime, SbS performs worse with larger $q$, showing no successes for $n_{min} \le 100$, whereas SRS continues to perform well for small $n_{min}$.

\subsection{Analytic Guarantees}
In this section, we evaluate the sufficient conditions of the lemmas in Sections~\ref{sec:ProposedURS} and \ref{sec:ProposedSbSSRS} by comparing against numerical results.
We focus on two critical steps in the proposed algorithm: the sampling step and the retrieval step.

For the sampling process, we first demonstrate the improvement in sampling probability afforded by SbS.
In Fig.~\ref{fig:Analysis}(a), we calculate and plot the URS and SbS sampling probabilities for a graph with the following parameters: $r \!=\! 2$, $\nmin \!=\! n_1 \!=\! 1000$, $p \!=\! 0.8$, $\obsprob \!=\! 0.7$, $n_2 \!=\! N \!-\! \nmin$, $q \!=\! 0.1 \balRS^{-1}$.
This scaling of $q$ puts the graph into the \smallq/ regime.
Along the x-axis, the size of the graph increases while the smallest cluster size remains constant.
For comparison, we calculate and plot the lower bound on $\clustsampprob_{min}$ from Lemma~\ref{lem:SbS_nsamples}, i.e. the right hand side of \eqref{eqn:9024099}.
Both the bound and the actual $\clustsampprob_{min}$ remain almost constant, while the cluster sampling probability of URS decreases as the graph becomes more imbalanced.
Fig.~\ref{fig:Analysis}(b) uses the same graph as (a), but shows the number of samples required to obtain at least $\nminbound \!=\! 10$ sampled nodes from each cluster.
To obtain this plot, we sample one node at a time from the full graph until $n_i' \ge \nminbound$ for all $1 \!\le\! i \!\le\! r$.
Once this threshold is attained, the sampling stops and the number of samples is reported as the sufficient $N'$ in the plot.
We perform this test for both URS and SbS, and plot the corresponding sufficient conditions from the lemmas.
For URS, this bound is the right hand side of \eqref{eqn:Suff_Optimality_Partial_RS_uniform} from Lemma~\ref{lem:Suff_Optimality_Partial},
and for SbS, this bound is the right hand side of \eqref{eqn:Sbs_samples_suff} from Lemma~\ref{lem:Suff_Optimality_Partial_SbS}.
The bounds trend the same as the numerical results, although the plot suggests that the sufficient conditions could be tightened.

For the retrieval process, we first generate the characteristic vectors for a sketch having three clusters with sizes $n_1' \!=\! n_2' \!=\! \nminp$ and $n_3' \!=\! 10^5 \!-\! 2 \nminp$.
We then perform the retrieval step against a full graph generated from an SBM with cluster sizes $\nmin \!=\! n_1 \!=\! n_2 \!=\! 10^4$ and $n_3 \!=\! 10^6 \!-\! 2 \nmin$, and plot the number of successes in 20 trials.
In Fig.~\ref{fig:Analysis}(c), we set $\obsprob = 0.7$ and vary $\densdiff = p-q$.
\begin{figure}
\centering
\includegraphics[scale=1]{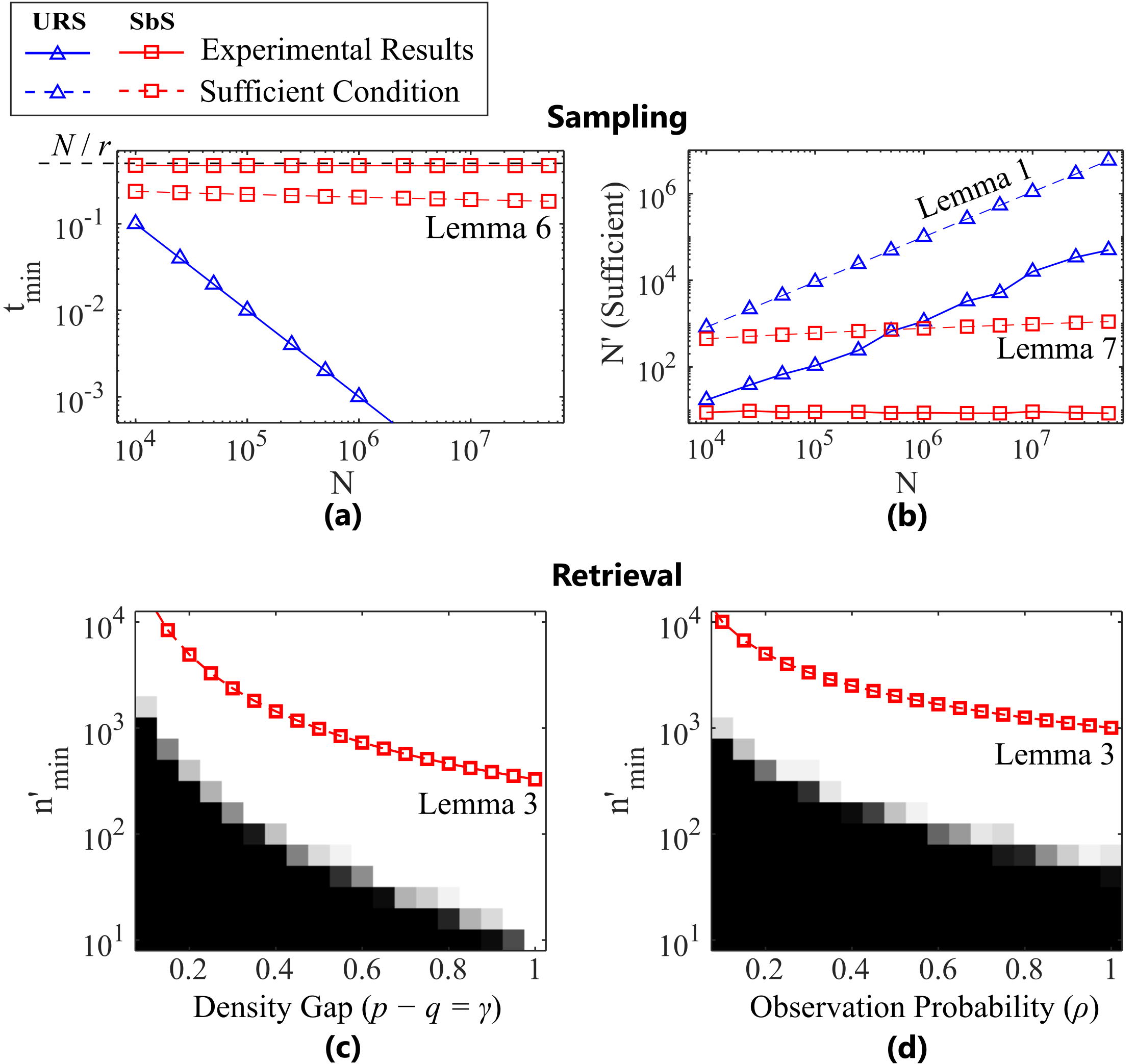}
\vspace{-.65cm}
\caption{Simulations supporting the lemmas for (a) sampling probability, (b) sufficient number of samples, and (c)-(d) sufficient conditions for retrieval.}
\label{fig:Analysis}
\end{figure}
In Fig.~\ref{fig:Analysis}(d), we set $\densdiff = p-q = 0.4$ and vary $\obsprob$.
The sufficient condition, found in the right hand side of \eqref{eqn:retrieval_suff} from Lemma~\ref{lem:Retrieve} is overlaid on top of the plot.
The region to the upper-right of the line is the region where success is guaranteed whp, and we can see that the success rate is indeed high in this region for both plots.

\subsection{Invoking Arbitrary Algorithms for Sketch Clustering}

As mentioned earlier, we can invoke any clustering algorithm in Step 2.1 of Algorithm~\ref{alg:mainalgorithm}.
In fact, the asymptotic results presented in Table~\ref{tab:AsymptoticComparison} indicate that the algorithms of \cite{cai2015,NIPS2016_6574} should work better than \cite{Chen:2014:CPO:2627435.2670322}, if they are used in the clustering step. 
Here, we repeat the simulations of Fig.~\ref{fig:FullVsSamp}(a) and (c), using both the SbS and SRS sampling methods, except here we use \cite{cai2015,NIPS2016_6574,Chen:2014:CPO:2627435.2670322} in step 2.1 of Algorithm~\ref{alg:mainalgorithm}.
\begin{figure}
\centering
\includegraphics[scale=1]{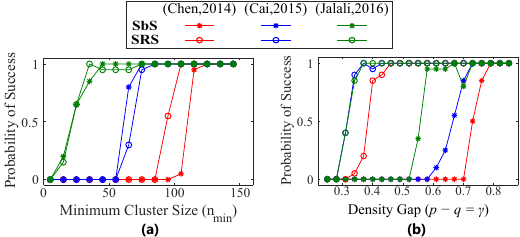}
\vspace{-.65cm}
\caption{Example of proposed framework invoking (Cai,2015) and (Jalali,2016)}
\label{fig:SamplieCaiJalali}
\end{figure}
Fig. \ref{fig:SamplieCaiJalali} shows that \cite{cai2015,NIPS2016_6574} indeed exhibit improved results for a given sampling technique over \cite{Chen:2014:CPO:2627435.2670322} in these regimes.

\subsection{Real-world Data and Graphs with Large Degree Variations}
\label{sec:experimentalresults}

The results presented up to this point have all used synthetic data.
Now we consider the application of the proposed framework to community detection in real data.
We must first address a challenging aspect of real-world data, which is the large variation and skew typically found in the degree distribution of their graphs.
This degree profile usually cannot be captured accurately by the SBM, thus motivating the creation of better fitting extensions such as the popular Degree Corrected SBM (DCSBM) \cite{PhysRevE.83.016107}.
Likewise, without modification, many algorithms which work well on the SBM fail to properly cluster graphs with skewed degree distributions.
Our framework is no exception.
To counteract the detrimental effects of the degree variations on SbS, we make a heuristic modification to this sampling method.
Specifically, rather than sampling node $j$ proportionally to $\frac{1}{\norm{\ba_j}_0}$, we sample proportionally to $\frac{\norm{\ba_j}_0}{\sum_{k = 1}^N a_{jk} \norm{\ba_k}_0}$, i.e. inversely proportional to the average degree of the node and its neighbors.

We first provide a result using synthetic data generated by the DCSBM.
Each node $i$ has a free parameter $\theta_i$ associated with it to influence the degree of the node. The probability of an edge between two nodes $i,j$ is $\theta_i \theta_j p$ if the nodes are in the same cluster, and $\theta_i \theta_j q$ otherwise. Following the examples presented in \cite{chen2018}, we draw each $\theta_i$ independently from a Pareto distribution. Specifically, this is the distribution having probability density function $f(x) \!=\! \alpha \beta^\alpha / x^{\alpha+1}$, where $\beta \!=\! (\alpha\!-\!1)/\alpha$.
Here, rather than invoking the low rank plus sparse decomposition algorithm \cite{Chen:2014:CPO:2627435.2670322}, we invoke the Convexified Modularity Maximization (CMM) algorithm which copes well with DCSBM graphs \cite{chen2018}.
Results are shown in Table~\ref{tab:RealWorldData}, with the average and standard deviation calculated over 20 trials.
\begin{table*}[t!]
  \centering
  \caption{\bf Results with graphs having skewed degree distributions (mean $\pm$ standard deviation) }
  \begin{tabular}{c c c c cc ccc c}
  \hline
  \multirow{2}{*}{Dataset} & \multirow{2}{*}{$N$} & \multirow{2}{*}{$N'$} & \multirow{2}{*}{$\frac{N}{\nmin}$} & \multicolumn{2}{c}{Full (CMM)} &  \multicolumn{3}{c}{Averaged SbS (CMM)} & URS (CMM) \\
  \cmidrule(lr){5-6} \cmidrule(lr){7-9} \cmidrule(lr){10-10}
  & & & & \# Errors & Time (s) & $\frac{N'}{\nminp}$ & \# Errors & Time (s) & \# Errors \\
  \hline
    DCSBM (Balanced, $\alpha \!=\! 2$)     & $7000$ & $1000$ & $3$          & $0.05 \pm 0.22$         & $451$ & $3.2 \pm 0.1$ & $0$                    & $8.0$ & $0$ \\
    DCSBM (Unbalanced, $\alpha \!=\! 2$)   & $7000$ & $1000$ & $7$ & $1662 \pm 400$ & $457$          & $4.5 \pm 0.2$ & $0.4 \pm 1.6$ & $8.2$          & $1519 \pm 607$ \\
    DCSBM (Unbalanced, $\alpha \!=\! 1.7$) & $7000$ & $500$  & $7$ & $1821 \pm 109$          & $476$          & $4.9 \pm 0.6$          & $594 \pm 863$          & $2.5$          & $2450 \pm 1167$ \\
    \hline
    Political Blogs (Full)        & $1222$       & $700$  & $2.1$                  & $61$                  & $9.4$ & $2.2 \pm 0.1$          & $72 \pm 11$           & $3.2$ & $102 \pm 16$ \\
    Political Blogs (Unbalanced)  & $747 \pm 12$ & $200$  & $4.6 \pm 0.2$ & $226 \pm 39$ & $3.7$          & $2.9 \pm 0.3$ & $152 \pm 62$ & $0.9$          & $232 \pm 66$ \\
    \hline
  \end{tabular}
  \label{tab:RealWorldData}
\end{table*}
In the first row, we generate a DCSBM graph with $\alpha\!=\!2$, $p=0.1$, $q=0.05$, and with three equal-sized clusters (within a rounding error).
All methods yield almost error-free estimates, and over a 56-fold speedup with both URS and our averaged SbS variant over full graph clustering.
In the next two rows, we show results for an unbalanced graph, having the same $p$ and $q$, but with one large cluster of size 5000, and two small clusters of size 1000 each.
When $\alpha\!=\!2$, SbS reduces the imbalance in the full graph from $\frac{N}{\nmin} \!=\! 7$ to $\frac{N'}{\nminp} \!\approx\! 4.5$ in the sketch, leading to a corresponding reduction in error. The error is still reduced significantly in the third row for the more difficult case of $\alpha\!=\!1.7$ with half as many samples.

In the fourth row of the table, we show experimental results for the popular Political Blogs dataset \cite{Adamic:2005:PBU:1134271.1134277}. We note that this dataset has two clusters of similar size.
We find that sampling 700 nodes with the SbS variant gives a comparable error rate to that of full graph clustering (though slightly higher), while having a nearly three-fold speedup.
To test the balancing effect of our sampling, for the fifth row we derive a new dataset from the Political Blogs graph: we retain all blogs from the liberal cluster, but sample 200 blogs from the conservative cluster uniformly at random, and then keep the largest connected component of the induced subgraph.
In addition to a four-fold speedup, the SbS variant improves the balance of the sketch graph, and we see a marked reduction in the errors.

\section{Conclusion} \label{sec:conclusion}

We proposed a sketch-based framework for clustering graphs.
The first goal was to improve computational efficiency over existing algorithms.
This was accomplished by sampling nodes to produce a small sketch, clustering the sketch using an existing clustering algorithm, and then inferring the clusters in the full graph based on the sketch clusters.
The speed of the framework was demonstrated in the case that nodes were sampled uniformly at random.
The second goal was to improve the handling of small clusters, while still maintaining low computational complexity.
This was accomplished by two sampling methods: one samples based on node degree (SbS), and the other performs spatial sampling (SRS).

\appendices

\section{Proofs for Random Sampling Approach} \label{sec:ProofSRS}

In this appendix, we provide proof of the lemmas and theorems for URS presented in Section~\ref{sec:ProposedURS}.
The set $\clusteridx_i$ contains the indices of the columns which belong to cluster $i$, and the set $\blockidx{i} = \clusteridx_i \cap \calI$ contains the indices of sampled columns which belong to cluster $i$.
\begin{proof}[Proof of Lemma~\ref{lem:Suff_Optimality_Partial}]

To simplify the analysis, we treat the diagonal elements of the adjacency matrix as regular intra-cluster connections, i.e. a self-loop occurs with probability $p$ and is observed with probability $\obsprob$. This modification can only hurt performance, and therefore provides a valid lower bound.
Additionally, as the graph size increases, the difference in performance will be negligible.

For simplicity, we will perform the analysis on a Bernoulli sampling model in place of the uniform sampling model.
In the Bernoulli model, a Bernoulli trial with success probability $N'/N$ is repeated for each node, and a node is included if its trial is successful.
Let $\Omega'$ be the index set of columns sampled using the Bernoulli model, and $\Omega$ be the index set of columns sampled using URS (without replacement).
Although the exact number of samples will vary around a mean of $N'$ in the Bernoulli model, from \cite{1580791}, we can conclude that
\begin{align}
\bbP \left( \nminp < \nminbound \cond \calI = \Omega' \right)
\ge \frac{1}{2} \bbP \left( \nminp < \nminbound \cond \calI = \Omega \right).
\end{align}
In words, the probability of failing to sample sufficient columns using URS is no more than twice the probability of failing using Bernoulli random sampling.

Now, in the Bernoulli model the number of samples $n_i'$ from cluster $i$ is a Binomial random variable with $n_i$ independent experiments, each with success probability $N'/N$.

Let $\xi_i, i = 1, \ldots, r$ be such that the number of samples $N' = \xi_i \nminbound N / n_i$.
Using the Chernoff bound for Binomial distributions \cite{McDiarmid1998} and the union bound, 
\begin{align} \label{eqn:95145}
\bbP\left(\nminp \ge \nminbound \cond \calI = \Omega \right)
&\ge 1 - 2 \sum_{i=1}^r \bbP \left( n_i' < \nminbound \cond \calI = \Omega' \right)
\nonumber \\
&\ge 1 - 2 \sum_{i=1}^r \exp \left( -\frac{\nminbound (\xi_i-1)^2}{2 \xi_i} \right).
\end{align}
Thus, if \eqref{eqn:Suff_Optimality_Partial_RS_uniform} holds, then $\xi_i \ge 2 + \frac{2}{\nminbound} \log \left(2rN \right)$ for $1 \le i \le r$ and the RHS of \eqref{eqn:95145} is lower-bounded by $1-N^{-1}$.

The upper bound in \eqref{eqn:Suff_Optimality_Partial_RS_uniform} is required since we are sampling without replacement, and the sketch matrix cannot be larger than the full matrix.

\end{proof}

\begin{proof}[Proof of Lemma~\ref{lm:completion}]

First we consider the lower bound in \eqref{eq:lemm2_suffRS}.
We need to sample a sufficient number of columns such that the decomposition of the sketch will be successful whp.
Let $\nminbound = \frac{\sqrt{C N'} \log N'}{\sqrt{\obsprob} \densdiff}$.
Successful decomposition is guaranteed by Theorem 4 of \cite{Chen:2014:CPO:2627435.2670322} with probability $1 - c N'^{-10}$
if $\nminp \ge \nminbound$.
If sufficient conditions \eqref{eqn:21973} and \eqref{eq:lemm2_suffRS} hold, then
\begin{align} \label{eqn:completion2c}
N \ge N' \ge 2 \balRS \left[ \frac{\sqrt{C N'} \log N'}{\sqrt{\obsprob} \densdiff} + \log \left( 2rN \right) \right],
\end{align}
and therefore Lemma~\ref{lem:Suff_Optimality_Partial} guarantees that $\nminp \ge \nminbound$ with probability at least $1-N^{-1}$.

The cluster sizes in the sketch cannot exceed those in the full graph, so we need $\nmin \ge \nminbound$. This is satisfied by \eqref{eqn:2348798}, which ensures that $\nmin^2 \ge \zeta N'$.

\noindent

\end{proof}

\begin{proof}[Proof of Lemma~\ref{lem:Retrieve}]

Define the inner product between the $j\textsuperscript{th}$ column and the $i\textsuperscript{th}$ characteristic vector of $\bL'$ (i.e. the eigenvector of $\bL'$ representing cluster $i$) as $\colinner_i(j) = (\ba_j')^T \bv_i$.
Suppose that $\ba_j'$ belongs to the $\ell \textsuperscript{th}$ cluster.
Then, $\colinner_\ell(j) \sim \Bin(n'_\ell, \obsprob p)$ whereas $\colinner_i(j) \sim \Bin(n'_i, \obsprob q)$ for any $i \neq \ell$.

Now, we define $\colinnernorm_i(j) = \colinner_i(j)/n_i'$, which is the normalized inner product used in step 3 of Algorithm~\ref{alg:mainalgorithm}.
Let $\densavg=\frac{\obsprob (p+q)}{2}$ and note that $\bbE[\colinnernorm_\ell(j)] = p$ and $\bbE[\colinnernorm_i(j)] = q$ for $i \neq \ell$.
Then, from the upper and lower Chernoff bounds \cite{McDiarmid1998}, it follows that
\begin{align}
\bbP\left(\colinnernorm_\ell(j) \leq \tau \right)
&\le \exp\left(-\frac{\obsprob (p-q)^2}{8p}n_\ell' \right) ,\label{eqn:Supperbound}
\\
\bbP\left(\colinnernorm_i(j) \geq \tau \right)
&\le \exp\left(-\frac{3 \obsprob (p-q)^2}{24q + 4(p-q)}n_i' \right). \label{eqn:Slowerbound}
\end{align}
If \eqref{eqn:retrieval_suff} holds, then
\begin{align} \label{eqn:retrieval_suff2}
\nminp &\ge \frac{8p}{\obsprob (p-q)^2} \log\left(rN^2\right)
> \frac{4 (p+5q)}{3 \obsprob (p-q)^2} \log\left(rN^2\right)
\end{align}
(since $p>q$) and the right hand sides of \eqref{eqn:Supperbound} and \eqref{eqn:Slowerbound} will be upper bounded by $(rN^2)^{-1}$.

Then, the probability that the algorithm will fail to classify column $j$ is the probability that the normalized inner product is larger for an incorrect cluster than for the correct cluster,
\begin{align}
&\bbP\left(\bigcup_{\substack{i=1 \\ i \ne \ell}}^r \colinnernorm_\ell(j) \le \colinnernorm_i(j) \right)
\nonumber \\
&\le 
\bbP\left( \left[ \colinnernorm_\ell(j) \le \densavg \right] \bigcup \left[ \bigcup_{\substack{i=1 \\ i \ne \ell}}^r  \colinnernorm_i(j) \ge \densavg \right] \right)
\nonumber \\
&\le
\bbP\left( \colinnernorm_\ell(j) \le \densavg \right)
+ \sum_{\substack{i=1 \\ i \ne \ell}}^r \bbP\left( \colinnernorm_i(j) \ge \densavg \right)
= N^{-2}.
\end{align}
Thus, the probability of correctly classifying all columns is
\begin{align}
&\bbP\left(\bigcap_{j=1}^N \left[ \bigcap_{\substack{i=1 \\ i \ne \ell}}^r \colinnernorm_\ell(j) > \colinnernorm_i(j) \right] \right)
\nonumber \\
&\ge 1 - \sum_{j=1}^N \bbP\left(\bigcup_{\substack{i=1 \\ i \ne d}}^r \colinnernorm_\ell(j) \le \colinnernorm_i(j) \right)
> 1 - N^{-1}.
\end{align}

\end{proof}

\begin{proof}[Proof of Theorem~\ref{thm:maintheoremRS}]

First, let $\nminbound = \frac{ 8 p }{\obsprob \densdiff^2} \log\left(rN^2\right)$.
Conditions \eqref{eqn:2348798} and \eqref{eqn:8923478789} imply that
\begin{align} \label{eqn:9234689}
N \ge N' &\ge 2 \balRS \left[\frac{8p}{\obsprob \densdiff^2} \log\left( r N^2 \right) + \log \left(2 r N \right)\right],
\end{align}
and so Lemma~\ref{lem:Suff_Optimality_Partial} guarantees that $\nminp \ge \nminbound$ with probability at least $1-N^{-1}$.
If the sufficient condition of Lemma~\ref{lem:Suff_Optimality_Partial} is met, then retrieval is also guaranteed with probability $(1-N^{-1})$ by Lemma~\ref{lem:Retrieve}.
Additionally, for retrieval to be possible, we need $\nmin \ge \nminbound$ (i.e., the lower bound on cluster size cannot exceed the smallest cluster size in the full graph), which is satisfied by \eqref{eqn:8987239}.

Next, conditions \eqref{eqn:2348798} and \eqref{eqn:8923478789} also satisfy the conditions of Lemma~\ref{lm:completion}, thus guaranteeing successful decomposition with probability $1 \!-\! c {N}^{-10} \!-\! N^{-1}$.

Therefore, the complete clustering algorithm will be successful with probability at least $1 \!-\! c N^{-10} \!-\! 3N^{-1}$.

\end{proof}

\section{Proofs for Sparsity-based Sampling Approach} \label{sec:ProofsSbS}

In this appendix, we provide proofs related to the SbS sampling approach.
First, we present some technical lemmas which will be used in the remainder of this section.  We use the expected degree for a column in cluster $i$, which is $\mu_i = \obsprob [(p - q) n_i + q N]$.

We will often find it useful to work with a ``typical'' graph,
i.e., one whose adjacency matrix has all node degrees close to their respective mean values.
Specifically, define the set of typical graphs as
\begin{align}
\Atyp{\epsilon} = \left\{ \bA \cond \abs{ \norm{ \ba_j }_0 - \mu_i } \le \epsilon \, \mu_i \:,\: 1 \le i \le r, \, j \in \clusteridx_i \right\},
\end{align}
for an arbitrary $\epsilon>0$.
The following lemma bounds the probability with which the SBM will generate a typical graph.
\begin{lemma}[Typical graph] \label{lem:typgraph}
\label{lem:SbS_typical}
Given an arbitrary $\delta>0$, let
\begin{align}
1
> \epsilon
&\ge \sqrt{ \frac{3}{\mumin}\log\left(\frac{2 N}{\delta}\right) },
\end{align}
then $\prob\left( \bA \in \calA_\epsilon \right) \ge 1-\delta$.

\end{lemma}

\begin{proof}[Proof]

From the Chernoff bound \cite{McDiarmid1998}, if $\mumin \ge \frac{3}{\epsilon^2}\log\left(\frac{2N}{\delta}\right)$, then for a given column $j$ belonging to cluster $i$,
\begin{align}
\bbP\left( \abs{ \norm{ \ba_j }_0 - \mu_i } \le \epsilon \, \mu_i \right)
\ge 1 - \frac{\delta}{N}.
\end{align}
Finally, we take the union bound over all $N$ columns.

\end{proof}

Next, we will place bounds on the sampling probabilities obtained using SbS for a typical graph.
\begin{lemma}[SbS column sampling probabilities]
\label{lem:SbS_nsamples_col}

Let $\colsampprob(j)$ be the probability of sampling a single column $j$, assuming that we are sampling using SbS with replacement.
Furthermore, suppose that $\InAtyp{\suffepsA}$ with $\suffepsA$ as defined in \eqref{eqn:95422}. Then for an arbitrary column $j$ belonging to cluster $i$, $\colsampprob(j) \ge \colsampprob_i^-$, where $\colsampprob_i^{-} = \frac{1-\suffepsA}{1+\suffepsA} \frac{1}{ r n_i \eta_i}$ and $\eta_i = \left[ 1 + \frac{q}{p} \left(\frac{N}{n_i}-1\right) \right]$.

\end{lemma}

\begin{proof}[Proof]

Since $\bA \in \Atyp{\suffepsA}$, then from Lemma~\ref{lem:typgraph},
\begin{align} \label{eqn:98246987}
\colsampprob(j)
& = \left( \sum_{k=1}^N \frac{\norm{\ba_j}_0}{\norm{\ba_k}_0} \right)^{-1}
\ge \left( \sum_{k=1}^r \frac{(1+\suffepsA) \mu_i n_k}{(1-\suffepsA) \mu_k} \right)^{-1}
\nonumber \\
& = \frac{1-\suffepsA}{1+\suffepsA} \left( \sum_{k=1}^r \frac{[n_i + \frac{q}{p} (N-n_i)] n_k}{[n_k + \frac{q}{p} (N-n_k)]} \right)^{-1}
\ge \colsampprob_i^{-}.
\end{align}

\end{proof}

Next we provide proofs for the SbS lemmas and theorems found in Section~\ref{sec:Sbs}.

\begin{proof}[Proof of Lemma~\ref{lm:sbs_1}]

Because $\bA=\bL$, the degree of a node belonging to cluster $i$ is exactly $n_i$.
Then the probability of sampling from cluster $i$ is $\frac{ n_i \frac{1}{ n_i} }{ \sum_{j=1}^r n_j \frac{1}{ n_j}} = \frac{1}{r}$.

\end{proof}

\begin{proof}[Proof of Lemma~\ref{lem:SbS_nsamples}]

Invoking Lemma~\ref{lem:SbS_typical}, the probability that $\InAtyp{\suffepsA}$ is at least $1-N^{-1}$.
In this event, from Lemma~\ref{lem:SbS_nsamples_col} we have 
\begin{align}
\clustsampprob_{min}
&=  \min_{1 \le i \le r} \sum_{j \in \clusteridx_i} \colsampprob(j)
\ge \min_{1 \le i \le r} n_i \colsampprob_i^{-},
\end{align}
which yields \eqref{eqn:9024099}.

\end{proof}

\begin{remark} \label{rem:SbsSamp3}
In the limit as $p-q \rightarrow 0$, we would expect \eqref{eqn:9024099} to become $\clustsampprob_{min} \ge \frac{1-\suffepsA}{1+\suffepsA} \balRS^{-1}$ such that it is consistent with the column sample probabilities for URS. Indeed, in this case $q/p \rightarrow 1$, and \eqref{eqn:98246987} becomes $\colsampprob_i \ge \frac{1-\suffepsA}{1+\suffepsA} N^{-1}$, which is the desired result.
\end{remark}

\begin{proof}[Proof of Lemma~\ref{lem:Suff_Optimality_Partial_SbS}]

For ease of analysis, we calculate the sufficient number of samples \emph{with} replacement to obtain $\nminbound$ \emph{distinct} columns from each cluster.
This condition will also be sufficient for sampling without replacement, since fewer samples will be required than sampling with replacement.

First, we will bound the probability that at least $n_i-b$ columns are \emph{not} sampled from cluster $i$ in a particular sketch of $N'$ samples.
In the following, let $\colsampprob(k)$ be the probability of sampling column $k$, $m(k)$ be the number of times that column $k$ is sampled, and $\Atyp{\suffepsB}$ be the typical graph as defined in Lemma~\ref{lem:SbS_typical} with $\epsilon=\suffepsB$.
From \eqref{eqn:43453} and Lemma~\ref{lem:SbS_typical}, we have that $\prob(\NotInAtyp{\suffepsB}) < (2rN)^{-1}$.

Let $\calS_i$ be a set containing exactly $n_i-b$ distinct column indices from cluster $i$.
The probability that the columns in $\calS_i$ are \emph{not} sampled by SbS is $\prob \left( \sum_{k \in \calS_i } m(k) = 0 \right) = \left( 1-\sum_{k \in \calS_i } s(k) \right)^{N'}$.
Note that this probability includes the event that other columns not in $\calS_i$ are also absent from the sample.
If the graph belongs to a typical set, then
\begin{align}
    \prob \left( \sum_{k \in \calS_i } m(k) = 0 \cond \InAtyp{\suffepsB} \right)
    &\le (1-(n_i - \nminbound) s_i^-)^{N'}
    \nonumber \\
    &\le \exp \curly{ -(n_i - \nminbound) N' s_i^- }.
\end{align}
Since there are $\binom{n_i}{b}$ ways to choose the elements in set $\calS_i$, the probability that \emph{fewer} than $b$ distinct columns are sampled from cluster $i$ is
\begin{align} \label{eqn:876832}
    \prob \left( n_i' < b \cond \InAtyp{\suffepsB} \right)
    &\le \binom{n_i}{b} \prob \left( \sum_{k \in \calS_i} m(k) = 0 \cond \InAtyp{\suffepsB} \right)
    \nonumber \\
    &\qquad \le n_i^b \left\{ \exp \bracket{ -(n_i-b) N' s_i^- } \right\}.
\end{align}
If \eqref{eqn:Sbs_samples_suff} holds, then the RHS of \eqref{eqn:876832} is less than $(r N)^{-1}$.
Then, the probability of failure for any graph is
\begin{align}
    \prob \left( n_i' < b \right)
    &\le \prob \left( n_i' < b \cond \InAtyp{\suffepsB} \right)
         + \prob \left( \NotInAtyp{\suffepsB} \right)
    \le (rN)^{-1}.
\end{align}

Finally, applying the union bound over all clusters, we have $\bbP\left(\nminp \ge \nminbound \right) = 1 - \bbP\left\{ (n_1' < \nminbound) \cup  \cdots \cup (n_r' < \nminbound) \right\}$,
which is greater than or equal to $1 - \sum_{i=1}^r \bbP\left(n_i' < \nminbound \right) = 1 - N^{-1}$.
\end{proof}

\begin{proof}[Proof of Theorem~\ref{thm:sketchnmin}]

We will invoke Lemma~\ref{lem:Suff_Optimality_Partial_SbS} with $b = \frac{\sqrt{C N'} \log N'}{\sqrt{\obsprob'} \densdiff'}$.
From \eqref{eqn:786867123}, it follows that $\nminbound \le \nmin/2$. Critically, this ensures that $\left( 1-\frac{\nminbound}{\nmin} \right)$ will be bounded away from zero, and will impose the bound $g' \ge g$.

Then, from conditions \eqref{eqn:786867123} and \eqref{eq:lemm2_suffSbSA}, and since $g' \ge g$, condition \eqref{eqn:Sbs_samples_suff} is satisfied.  Therefore, Lemma~\ref{lem:Suff_Optimality_Partial_SbS} guarantees \eqref{eqn:9893784} to hold with probability at least $1-N^{-1}$.

\end{proof}

\begin{proof}[Proof of Theorem~\ref{thm:sketchpqrho}]
Because we need to calculate separately the edge probabilities $p', q'$ and observation probability $\rho'$ probability, we will need to separately consider the fully observed adjacency matrix, denoted $\fullobsmat$, and the indicator matrix for the observed entries, denoted $\obsmat$. The matrix $\obsmat$ contains zeros for unobserved entries and ones for observed entries. Both $\fullobsmat$ and $\obsmat$ have zeros along the diagonal so that the partially observed adjacency matrix $\bA$ is the sum of the identity matrix with the element-wise product of $\fullobsmat$ and $\obsmat$.
Define
$\AOneFull = \left\{ \fullobsmat \cond \norm{\fullobsvec_j}_0 \ge 1 \:,\: 1 \le j \le N \right\}$ and 
$\AOneObs = \left\{ \obsmat \cond \norm{\obsvec_j}_0 \ge 1 \:,\: 1 \le j \le N \right\}$.
For brevity, we use $\InAOneFull$ as shorthand for the event $\fullobsmat \in \AOneFull$, and $\NotInAOneFull$ as shorthand for the event $\fullobsmat \notin \AOneFull$, and likewise for $\obsmat$ and $\AOneObs$.
We define $\fullobsel_{ij}$ as the element of matrix $\fullobsmat$ at the $i\textsuperscript{th}$ row and $j\textsuperscript{th}$ column, and similarly for matrix $\obsmat$.

We define $\pqprimeeps$ as the probability that $\AOneFull$ is atypical.
Then,
\begin{align}
\pqprimeeps
&= \bbP\left(\fullobsmat \notin \AOneFull \right)
= \bbP\left\{ \cup_{j=1}^{N} \norm{\fullobsvec_j}_0 = 0 \right\}
\nonumber \\
&\le \sum_{i=1}^N \bbP\left(\norm{\fullobsvec_j}_0 = 0 \right)
= \sum_{i=1}^r n_i (1 - p)^{n_i} (1 - q)^{N-n_i}
\nonumber \\
&\le N (1-p)^{\nmin} (1-q)^{\nmin}.
\end{align}
Likewise,
\begin{align}
\rhoprimeeps
&= \bbP\left(\obsmat \notin \AOneObs \right)
= \bbP\left\{ \cup_{j=1}^{N} \norm{\obsvec_j}_0 = 0 \right\}
\nonumber \\
&\le \sum_{i=1}^N \bbP\left(\norm{\obsvec_j}_0 = 0 \right)
= N (1-\obsprob)^N.
\end{align}

We will first bound $p'$. Let the set of intra-cluster edges contained in the sketch matrix be $\calP = \left\{(\samprand_i,\samprand_j) \in \bigcup_{k=1}^r (\calI_k \crossprod \calI_k) \cond \samprand_i \ne \samprand_j \right\}$. We denote the $i\textsuperscript{th}$ and $j\textsuperscript{th}$ samples as $S_i$ and $S_j$, respectively. 
Then, the probability that an arbitrary intra-cluster element of the fully observed sketch matrix $\fullobsel_{ij}'$ contains a one is as follows.

For the upper bound, we have
\begin{align} \label{eqn:4321213}
p'
&= \bbP\left( \fullobsel_{\samprand_i \samprand_j} = 1 | (\samprand_i,\samprand_j) \in \calP \right)
\nonumber \\
&= \frac{ \sum_{(\samp_i,\samp_j) \in \calP} \bbP\left( \samp_i,\samp_j \cond \fullobsel_{\samp_i \samp_j} = 1 \right) \bbP\left( \fullobsel_{\samp_i \samp_j} = 1 \right) }
{ \sum_{(\samp_i,\samp_j) \in \calP} \bbP\left( \samp_i,\samp_j \right)}
\nonumber \\
&= p \: \frac{ \bbP\left( \samp_i,\samp_j \cond \fullobsel_{\samp_i \samp_j} = 1 \right) }
{ \bbP\left( \samp_i,\samp_j \right)}
\le p.
\end{align}
In the last line, we used the fact that the presence of an edge slightly reduces the probability of sampling the respective columns.

For the lower bound, first observe that
\begin{align}
\bbP\left( \fullobsel_{\samp_i\samp_j}' = 1 | \InAOneFull \right)
&= \frac{p - \bbP\left( \fullobsel_{\samp_i,\samp_j} = 1 \cond \NotInAOneFull \right) \bbP\left( \NotInAOneFull \right)}{\bbP\left( \InAOneFull \right)}
\nn\\
&\ge \frac{p - p \: \bbP\left( \NotInAOneFull \right)}{\bbP\left( \InAOneFull \right)}
\ge p(1-\pqprimeeps).
\end{align}
In the second line, we have used the fact that if the graph $\fullobsmat$ has some nodes with degree zero, then this slightly lowers the probability that a given element has a one, and therefore $\bbP\left( \fullobsel_{\samp_i,\samp_j} = 1 \cond \NotInAOneFull \right) < p$.
Then, we have
\begin{align} \label{eqn:983789243}
p'
&= \bbP\left( \fullobsel_{ij}' = 1 \cond (\samprand_i,\samprand_j) \in \calP \right)
\nonumber \\
&=\bbP\left( \fullobsel_{ij}' = 1 \cond \InAOneFull,(\samprand_i,\samprand_j) \in \calP \right) \bbP\left( \calA_1 \right)
\nonumber \\
&\qquad + \bbP\left( \fullobsel_{ij}' = 1 \cond \NotInAOneFull,(\samprand_i,\samprand_j) \in \calP \right) \bbP\left( \NotInAOneFull \right)
\nonumber \\
&\ge \bbP\left( \fullobsel_{ij}' = 1 \cond \InAOneFull,(\samprand_i,\samprand_j) \in \calP \right) (1 - \pqprimeeps).
\end{align}
Finally, 
\begin{align} \label{eqn:983789243_2}
&\bbP\left( \fullobsel_{ij}' = 1 \cond \InAOneFull,(\samprand_i,\samprand_j) \in \calP \right)
\nn\\
&= \frac{\sum_{(\samp_i,\samp_j) \in \calP} \bbP\left( \samp_i,\samp_j \cond \fullobsel_{\samp_i,\samp_j} = 1, \InAOneFull  \right) \bbP\left( \fullobsel_{\samp_i,\samp_j} = 1 \cond \InAOneFull \right)}
{\sum_{(s_i,s_j) \in \calP} \bbP\left( \samp_i, \samp_j | \InAOneFull \right)}
\nonumber \\
&\ge p (1 - \pqprimeeps).
\end{align}
Combining \eqref{eqn:983789243} and \eqref{eqn:983789243_2}, we arrive at the lower bound.

For bounding $q'$, we follow a similar line of reasoning as for $p'$, but substituting $\calQ = \{(\samprand_i,\samprand_j) \in \bigcup_{\stackrel{k,\ell=1}{k \ne \ell}}^r (\calI_k \crossprod \calI_\ell) \cond \samprand_i \ne \samprand_j \}$ in place of $\calP$, which yields \eqref{eqn:3242334}.
Likewise, to bound $\obsprob'$, we replace $\obsset = \left\{(\samprand_i,\samprand_j) \in (\calI \crossprod \calI) \cond \samprand_i \ne \samprand_j \right\}$ in place of $\calP$, and matrix elements $\obsel_{ij}$ in place of elements $\fullobsel_{ij}$, which yields \eqref{eqn:423423}.
\end{proof}

\begin{proof}[Proof of Theorem~\ref{thm:maintheoremSbS}]

Successful decomposition is guaranteed by Theorem 4 of \cite{Chen:2014:CPO:2627435.2670322} with probability $1 - c N'^{-10}$
if $\nmin'^2 \ge \frac{C N' \log^2 N' }{\obsprob' {\densdiff'}^2}$, which is satisfied by condition \eqref{eqn:34879800}.

Furthermore, we need to sample enough columns such that the retrieval process is successful whp.
We will invoke Lemma~\ref{lem:Suff_Optimality_Partial_SbS} with $b = \frac{ 8 p }{\obsprob \densdiff^2} \log\left(r N^2\right)$.  If conditions \eqref{eqn:421159} and \eqref{eq:lemm2_suffSbSB} hold, then $\nminbound \le \nmin/2$, $\balSbS' \ge \balSbS$, and
condition \eqref{eqn:Sbs_samples_suff} is satisfied.
Then, Lemma~\ref{lem:Suff_Optimality_Partial_SbS} guarantees that $\nminp \ge \nminbound$ with probability at least $1 \!-\! N^{-1}$.
If the sufficient conditions of Lemma~\ref{lem:Suff_Optimality_Partial_SbS} are met, then retrieval is guaranteed with probability $1 \!-\! N^{-1}$ by Lemma~\ref{lem:Retrieve}.

Using the union bound, all conditions hold with probability at least $1 \!-\! c N^{-10} \!-\! 2N^{-1}$.

\end{proof}

{\bibliography{bibfile}}
\bibliographystyle{IEEEtran}

\end{document}